\documentclass[10pt,journal,compsoc]{IEEEtran}

\usepackage{fancyhdr}
\usepackage{epsfig}
\usepackage{threeparttable}
\usepackage{epsf,epsfig}
\usepackage{amsthm}
\usepackage[cmex10]{amsmath}
\usepackage{amssymb, amsfonts}
\usepackage[noadjust]{cite}
\usepackage{dsfont}
\usepackage{subfigure}
\usepackage{color}
\usepackage{soul}
\usepackage{amssymb}
 \usepackage{accents}
  \usepackage{algorithm}
    \usepackage{url}
    \usepackage{framed} 
	\usepackage[colorlinks,
	linkcolor=black,
	anchorcolor=black,
	citecolor=black]{hyperref}
    \usepackage{caption2}
    \usepackage[english]{babel}
    \usepackage{multirow}
    \usepackage{threeparttable}

\newtheorem{lemma}{Lemma}

\newtheorem{proposition}{Proposition}

\newtheorem{assumption}{Assumption}

\newtheorem{remark}{\bf Remark}
\def\phi{\varphi}

\def\({\left(}
\def\){\right)}

\def\b0{{\mathbf{0}}}

\DeclareMathSizes{10}{9}{5}{3}

\def\rv#1{\textcolor{black}{#1}}
\usepackage{mathrsfs}
\usepackage{algorithm}
\usepackage{algorithmic}
\IEEEoverridecommandlockouts
\begin{document}	

\title{\huge  Joint Compression and Deadline Optimization for Wireless  Federated Learning  \vspace{-5pt}}


\IEEEtitleabstractindextext{%
\begin{abstract}
    \emph{Federated edge learning} (FEEL) is a popular distributed learning framework for privacy-preserving at the edge,    
    in which densely distributed edge devices periodically exchange model-updates with the server to complete the global model training. 
    Due to limited bandwidth and uncertain wireless environment, FEEL may impose heavy burden to the current communication system. 
    In addition, under the common FEEL framework, the server needs to wait for the slowest device to complete the update uploading before starting the aggregation process, 
    leading to the straggler issue that causes prolonged communication time. 
    In this paper, we propose to  accelerate FEEL from two aspects: i.e., 1) performing data compression on the edge devices and 2) setting a deadline on the edge server to exclude the straggler devices. 
    However, undesired {gradient compression errors} and {transmission outage}  are introduced by the aforementioned operations respectively, affecting the convergence of FEEL as well. 
    In view of  these practical issues, we formulate a training time minimization problem, with the compression ratio and deadline to be optimized. 
    To this end, 
    an asymptotically unbiased aggregation scheme is first proposed to ensure zero optimality gap after convergence, 
    and the impact of compression error and transmission outage on the overall training time are quantified through convergence analysis. 
    Then, the formulated problem is solved in an alternating manner, based on which, the novel \emph{joint compression and deadline optimization} (JCDO) algorithm is derived.  
    Numerical experiments 
    for different use cases in FEEL including image classification and autonomous driving 
    show that the proposed method is nearly 30X faster than the vanilla FedSGD algorithm, and outperforms the state-of-the-art schemes. 
\end{abstract}

}
\author{Maojun Zhang, Yang Li, Dongzhu Liu, \emph{Member, IEEE}, Richeng Jin*, \emph{Member, IEEE}, Guangxu Zhu*, \emph{Member, IEEE}, Caijun Zhong, \emph{Senior Member, IEEE}, and Tony Q.S. Quek, \emph{Fellow, IEEE} 
\IEEEcompsocitemizethanks{
    \IEEEcompsocthanksitem M. Zhang and C. Zhong are with the College of information Science and Electronic Engineering, Zhejiang University, Hangzhou, China. 
    \IEEEcompsocthanksitem Y. Li is with China Academy of Information and Communications Technology, Beijing, China. 
    \IEEEcompsocthanksitem D. Liu is with School of Computing Science, University of Glasgow. 
    \IEEEcompsocthanksitem R. Jin is with the Department of Information and Communication Engineering, Zhejiang University, Hangzhou, China, 310007, the Zhejiang–Singapore Innovation and AI Joint Research Lab, Hangzhou, China, 310007, and also with Zhejiang Provincial Key Lab of Information Processing, Communication, and Networking (IPCAN), Hangzhou, China, 310007 (e-mail: richengjin@zju.edu.cn). 
    \IEEEcompsocthanksitem G. Zhu is with the Shenzhen Research Institute of Big Data, Shenzhen, Guangdong 518172, China, and also with the Peng Cheng Laboratory, Shenzhen, Guangdong 518055, China, and also with the Pazhou Laboratory (Huangpu), Guangzhou, Guangdong 510555, China (e-mail: gxzhu@sribd.cn).  
    \IEEEcompsocthanksitem T. Q. S. Quek is with the Singapore University of Technology and Design, Singapore 487372.     
    }
}
\maketitle


\section{introduction}
\emph{Federated edge learning} (FEEL)  is a popular distributed learning framework for tackling the privacy-preserving training issues at the edge \cite{mcmahan2017communication}. 
By distributing the computations across edge devices, the edge server only needs to collect the local updates instead of the original private data. 
Through frequent model interaction, the rich data widely distributed at edge devices can be fully utilized for training without compromising their privacy. 
Nevertheless, the improvement of model performance in deep learning is usually accompanied with a significant increase  in model size. 
As the edge devices need to upload the computed learning updates (e.g., local gradients or models) to the edge server at each round, the data transmission brings a huge overload to the current wireless communication system, causing transmission jam or unbearable delay. 
This thus prompts an active research area focusing on developing communication-efficient techniques to deploy FEEL over wireless networks \cite{saad2019vision,zhu2020toward}. 
\subsection{Data Compression in FEEL}
To alleviate the communication bottleneck, a common and intuitive approach is to compress the local updates to be uploaded via lossy compression techniques, e.g., quantization and sparsification. 
By quantization, the compression is achieved through representing each element with less bits. A lot of quantization schemes had been developed for the distributed learning settings, e.g., SignSGD \cite{bernstein2018signsgd,zhu2020one,karimireddy2019error,zheng2020design}, TernSGD \cite{wen2017terngrad}, QSGD \cite{alistarh2017qsgd}, etc. 
By sparsification, the insignificant elements will be dropped, such that the compression can be achieved by only encoding the non-zero elements \cite{lin2017deep,stich2018sparsified,wangni2018gradient,zhang2021federated,isik2022sparse,wen2022federated,panda2022sparsefed,li2020ggs,ozfatura2021time}. 
The work \cite{konevcny2016federated} first considered the compression issue under the federated learning setting, where random mask, subsampling, and probabilistic quantization were proposed to reduce the uplink communication costs. 
Realizing the vector features of the local update, the universal vector based rather than element based quantization scheme was proposed to further reduce the communication overhead \cite{shlezinger2020uveqfed}. 
Based on quantization and sparsification, update compression had been achieved from the perspective of compressive sensing as well \cite{oh2021communication,jeon2020compressive}. 
With the aforementioned schemes, the transmission overhead could be reduced effectively. 
{In the meantime, {gradient compression error} will be introduced \cite{mitchell2022optimizing,liu2022hierarchical,reisizadeh2020fedpaq}. 
For conducting compression in practical FEEL scenario, in addition to designing compression algorithm, one should also consider setting the compression ratio to efficiently utilize the communication resources. 
}
Aware of the trade-off between the compression ratio and compression error, the authors in \cite{jhunjhunwala2021adaptive} proposed an adaptive quantization strategy, which starts with 
a more aggressive quantization scheme and gradually increases its precision
as training progresses. 
A similar ``latter is better’’ principle was proposed in \cite{shen2021resource}. 
{Besides, the authors in \cite{han2020adaptive} proposed a fairness-aware sparcification method, which holds an equal number of non zero elements among the updates from different devices. 
Nevertheless, the aforementioned methods focus on dynamic compression ratio over different training stages, 
while the unstable wireless environment, representing the communication ability, is neglected. 
Given this, it is desired to design the setting method for the purpose of combating against dynamic wireless channel and accelerating FEEL convergence. 
}

\subsection{Straggler Mitigation in FEEL}
In addition to performing data compression on the device side, some efforts can be made on the server side as well to tackle the communication bottleneck. 
In a typical FEEL system, the upload communication time depends on the slowest device. 
Due to the heterogeneity of computation and communication resources among devices, some devices with poor channel state or computational capacity may upload their updates much slower than the others, critically prolonging the per-round latency. 
This is the well-known ``straggler’’ problem \cite{chen2021distributed}. 
{
Given this, FEEL are extended to the asynchronous setting \cite{chai2021fedat,li2020efficient,xie2019asynchronous}, which enables devices to upload their updates periodically without strict synchronization. 
  However, the asynchrony usually results in divergent training, since the updates are obtained based on different global model and devices that can upload quickly tend to upload more frequently than those devices with slow upload speed. 
  Another alternative method for tackling straggler issue is device scheduling, which is applicable for synchronous FEEL.}
It suggests that the edge server should carefully choose the non-straggler devices to upload their model updates. 
On this basis, 
a device scheduling scheme was initially proposed in \cite{nishio2019client} to mitigate the straggler problem and sample devices based on their resource conditions. 
The authors in \cite{reisizadeh2020straggler} further proposed to start the federated learning process with the faster devices and gradually involve the stragglers to improve the model performance. 
In addition to active device sampling, a more practical and simpler way is to set a deadline for the uploading process to exclude the straggler devices \cite{reisizadeh2019robust,lee2021adaptive,nishio2019client}. 
The edge server will start the aggregation process after the deadline, the updates from the straggler devices are thus dropped. 
{
    The deadline can exclude the straggler devices effectively.  
    However, due to the unreliability of wireless environment, 
    even non-straggler devices may fail to deliver the local updates, 
    causing undesired   
    {transmission outages} \cite{chen2020joint} that hinder the convergence of the global model,  especially in the non independent and identically distributed (non-i.i.d.) data distribution settings. 
    The authors in \cite{wang2021quantized} jointly considered the impacts of data compression and transmission outage for compression ratio design, 
    and proposed to set the compression ratio to ensure equal outage probability across devices. 
    However, the deadline in \cite{wang2021quantized} is fixed and needs to be preset. 
    Given the transmission outage brought by the deadline, 
    it is desired to carefully set the deadline  for achieving communication-efficient and FEEL with fast convergence. 
    \rv{Furthermore, the likelihood of transmission outage is also affected by the amount of data being uploaded that depends on the compression ratio. 
    The collective influence on transmission outage probability arises a strong coupling between compression mechanism and the transmission deadline.} 
    However, 
    owing to the difficulty in quantifying the impacts of compression and deadline, prior works, to the best of our knowledge, focused on separate design (of either the deadline \cite{lee2021adaptive} or the compression ratio \cite{wang2021quantized,han2020adaptive}), or joint design with heuristic approaches \cite{reisizadeh2019robust}.  
    \rv{Unlike conventional hyperparameters (e.g., learning rate), the optimal values of compression ratio and deadline vary significantly across different learning tasks, wireless environments, and even training stages. Additionally, these parameters are interdependent, making manual tuning highly suboptimal and often leading to substantial performance degradation.
    Therefore, the joint design with theoretical guarantee remains a critical issue and warrants immediate attention.} 
    To close this research gap, we make the first attempt to formulate the joint design problem and solve it directly. 
    }

\subsection{Contribution and Organization}
In this paper, we consider a general FEEL system, which consists of multiple edge devices with heterogeneous communication capabilities and one edge server. At each round, the edge server needs to first determine a deadline for the local training and uploading process, and devices should compress the local update accordingly. 
We aim to formulate and solve the joint compression and deadline optimization problem with the objective of directly minimizing the total training time. 
The main contributions of this work are elaborated as follows.
\begin{itemize}
	\item{{\bf{Training time analysis:}} We present a communication-efficient FEEL system, where the devices compress the local update through sparsification and the server sets a deadline for the local update process to exclude the straggler devices. Realizing that the gradient compression error and transmission outage are side effect brought by the aforementioned two operations, 
    we proposed an asymptotically unbiased aggregation scheme.
    This is underpinned by a comprehensive theoretical analysis for their impact on convergence behavior, unveiling a complex interplay that significantly affects convergence speed. 
    }
	\item{{\bf{Training time minimization via joint optimization of compression ratio and deadline:}} Building upon the theoretical insights, we formulate the training time minimization problem by jointly optimizing the deadline and the sparsification ratio of the data compression. 
	The online optimization strategy is adopted to avoid non-causal information acquisition. 
    We solve the resulting problem in an alternating manner, i.e., optimizing one variable by considering the other to be given in each iteration. For the sub-problem of compression ratio optimization, the quasi closed-form  solution is derived for efficient calculation; For the sub-problem of deadline optimization, the optimal deadline can be obtained by bisection search efficiently, and it is found that the optimal deadline can adapt to the training stage and channel state through the dynamic change of a weight parameter. Iteratively solving the two sub-problems thus gives the proposed \emph{joint compression and deadline optimization} (JCDO) algorithm.}
	\item{{\bf{Performance evaluation:}} We conduct extensive numerical experiments to evaluate the performance of the proposed JCDO algorithm for different use cases in FEEL including image classification and autonomous driving. 
	It is shown that the proposed algorithm could greatly accelerate the vanilla FedSGD algorithm, and is robust against device heterogeneity (i.e., different computational capabilities or wireless channels) compared with the state-of-the-art schemes.}	
\end{itemize}

\emph{Organization:} The remainder of this paper is organized as follows. 
Section 2 introduces the learning and communication models. 
Section 3 gives the analysis of the errors introduced by compression and deadline setting, followed by the proposed unbiased aggregation scheme and convergence analysis.
Section 4 
formulates the training time minimization problem and proposes the JCDO algorithm to jointly optimize the compression ratio and deadline. 
Section 5 shows the experimental results, followed by concluding remarks in Section 6. 

\begin{figure*}[tt]
\centering
\includegraphics[width=14cm]{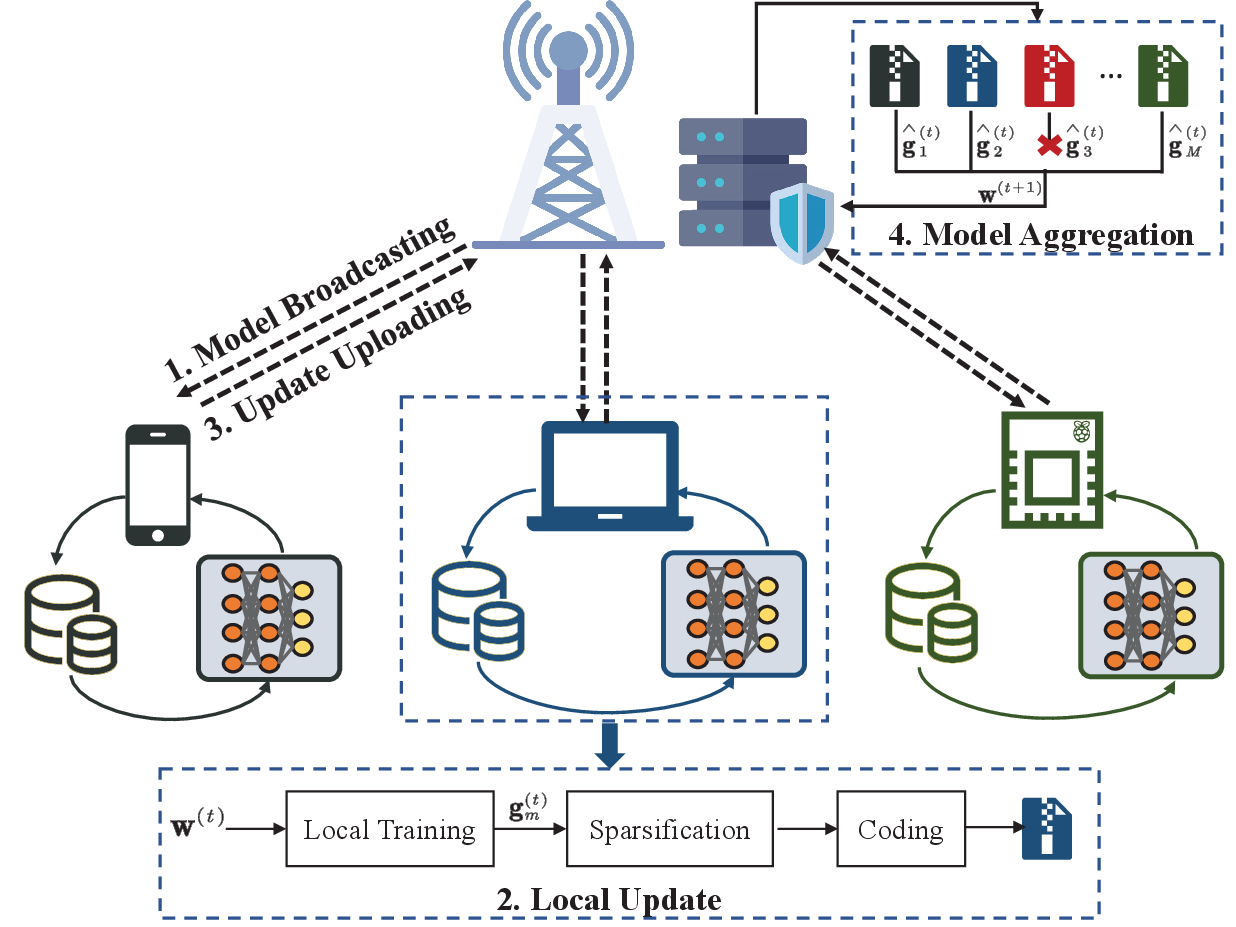}
\caption{Illustration of the federated edge learning system with stragglers}
\label{Fig:1}
\end{figure*}

\vspace{-2mm}
\section{System Models}\label{sec: system model}
We consider a synchronous \emph{ federated edge learning} (FEEL) system, as shown in Fig. \ref{Fig:1}, 
where there are one edge server and $M$ edge devices. 
The device set is denoted by $\mathcal{M}=\left\{1,2,3,...,M\right\}$. 
For each device $m$,  the private dataset comprises $d_m=\left|\mathcal{D}_m\right|$ pairs of training samples ${(\mathbf{x}_i,y_i)}_{i=1}^{d_m}$, where $\mathbf{x}_i$ is data vector and $y_i$ its label.  
We have global dataset $\left\{\mathcal{D}_m\right\}_{m=1}^M$ encompassing $d=\sum_{m=1}^M d_m$ training samples. 
\subsection{Learning Model}
The learning process is to minimize the global loss function in a distributed manner. 
Particularly, the global loss function on the entire distributed dataset is defined as 
\begin{align}\label{eq: the global loss function definition}
	L\left(\mathbf{w}\right) = \frac{1}{\sum_{m=1}^M d_m} \sum_{m=1}^M \sum_{\left(\mathbf{x}_j,y_j\right)\in \mathcal{D}_m} f\left(\mathbf{w};\mathbf{x}_j,y_j\right), 
\end{align}
where $\mathbf{w}$ denotes the global model, and $f\left(\mathbf{w};\mathbf{x}_j,y_j\right)$ is the sample loss quantifying the prediction error of $\mathbf{w}$ on the training sample $\left(\mathbf{x}_j,y_j\right)$. 
At each communication round $t$, 
we set a deadline $T_D^{(t)}$ to process the following steps and repeat them until the global model converges. 
\begin{itemize}
	\item{\bf{Global Model Broadcasting:}} The edge server broadcasts the current global model $\mathbf{w}^{(t)}$, and starts the countdown simultaneously. 
	\item{\bf{Local Model Training:}} Each device runs the \emph{stochastic gradient decent} (SGD) algorithm using its local dataset and the latest global model $\mathbf{w}^{(t)}$, and generates a local gradient estimate $\mathbf{g}_m^{(t)}$ as a surrogate of $\sum_{j=1}^{d_m}\nabla f\left(\mathbf{w}^{(t)};\mathbf{x}_j,y_j\right)$. 
	\item{\bf{Local Gradient Uploading:}} Each device compresses the local gradient as ${{\rm Comp} (\mathbf{g}_m^{(t)})}$, and then transmits it to the edge server. 
	\item{\bf{Global Model Updating:}} After the deadline, 
    the edge server update the global model by using the received gradients.  
\end{itemize}
{As illustrated above, the global model is distributed at the beginning of each round, then each device can improve the latest global model with the local dataset.\footnote{\rv{Note that our framework can be integrated with active device selection \cite{cui2022optimal,lai2021oort} by adjusting the model broadcasting process to only broadcast the global model to a selected set of devices. Subsequently, our deadline mechanism acts as a secondary selector, only considering devices that can successfully upload their update within the given deadline.}} Compared with the asynchronous scheme, the above method has a better convergence performance especially when tackling the non i.i.d. data distribution issue. 
}

In the Section \ref{subsec: compression model} and \ref{eq:subsec:unbiased aggregation}, we will detail the use of stochastic compression and the design of global update scheme respectively. 

\vspace{-2mm}
\subsection{Communication Model}\label{subsec: communication model}
The devices communicate to the server via \emph{orthogonal frequency-division multiple access} (OFDMA) channel. 
The available bandwidth is divided into $M$ sub-channels and each of them is allocated to a device.  
At each communication round $t$, each device uses the dedicated sub-channel $h_m^{(t)}$ with bandwidth $B$ to upload the compressed gradient ${\rm Comp}(\mathbf{g}_m^{(t)})$. 
We consider Rayleigh fading channel, i.e., $h_m^{(t)}\sim \mathcal{CN}(0, \sigma_m^2)$, where the channel coefficient is invariant within a communication round, and is \emph{independent and identically distributed} (i.i.d.) over the sequential rounds.
The achievable data rate of the channel between the user $m$ and the edge server is given by
\begin{align}\label{eq: data rate definition}
	C_m^{(t)} = B \log_2 \left(1 +\frac{P_m|h_m^{(t)}|^2}{BN_0}\right),  
\end{align} 
where $P_m$ is the transmit power and $N_0$ is the noise power spectral density.
\vspace{-4mm}
\subsection{Compression Model}\label{subsec: compression model}
\vspace{-1mm}
For a general vector $\mathbf{g}=[g_1,g_2,...,g_S]\in \mathbb{R}^S$, we compress $\mathbf{g}$ by 
an unbiased stochastic sparsification scheme as in \cite{wangni2018gradient}. 
Specifically, each element $g_i$ in $\mathbf{g}$ is preserved with probability $p_i$,  
and  
the sparsification function $\mathcal{S}\left(\cdot\right)$ is designed as  
\begin{align}\label{eq: definition of the compressor}
	\mathcal{S}\left(\mathbf{g}\right)=\left[Z_1\frac{g_1}{p_1},Z_2\frac{g_2}{p_2},...,Z_S\frac{g_S}{p_S}\right], 
\end{align}
where $Z_i$ is a \emph{random variable} (r.v.) following
Bernoulli distribution with mean $p_i$. 
We have a sparsity ratio $r$ as the expected ratio of elements to be preserved after compression.  
Note that, the gradient vector $\mathbf{g}$ usually has extremely high dimension in neural network. 
According to the \emph{law of large numbers} (LLN), we have $\lim_{S\rightarrow +\infty}\left\|{\rm Comp}(\mathbf{g},r)\right\|_0 /S = \sum_{i=1}^S p_i /S$. 
In the remaining of this paper, we consider the problem formulation in an asymptotic manner, with constraint on $\sum_{i=1}^S p_i/S \leq r$ due to the limited communication resources. 

The choice of $p_i$ is obtained by minimizing the variance of the compressed vector under the sparsity constraint as given below.
\begin{align}
	\begin{matrix}
		&\min\limits _{\left\{p_1,...,p_S\right\}}& \mathbb{E} \left\|\mathcal{S}\left(\mathbf{g}\right)-\mathbf{g}\right\|^2,\hfill \\
		&{\rm s.t.}  & \frac{\sum_{i=1}^S p_i}{S} \leq r, \hfill\\
		& &p_i \leq 1, \forall i. \hfill\\
	  \end{matrix}
	\end{align}
By using the \emph{Karush-Kuhn-Tucker} (KKT) conditions, we obtain the optimal solution  
$p_i^* = \min \{\frac{\left|g_i\right|}{\lambda},1\}$, where $\lambda$ is the Lagrange multiplier satisfying $\sum_{i=1}^S p_i^* = rS$. 

{After sparsification, 
the compression can be achieved through sparse vector coding. 
For a sparse vector $\mathbf{g}'=[g_1',g_2',...,g_S']$, we encode each non-zero element with its index and value, and upload the parameter set 
\begin{align}\label{eq:coding process}
	\mathcal{C}\left(\mathbf{g}'\right)= \left\{\left[i,g_i'\right]\vert g_i'\neq 0\right\}.
\end{align}
Therefore, for a given sparsity constraint $r$, using the aforementioned compression function ${\rm Comp}\left(\mathbf{g},r\right)=\mathcal{C}\left(\mathcal{S}\left(\mathbf{g},r\right)\right)$, the expected number of bits to be uploaded is $brS$, where $b$ is the number of bits for encoding a single non-zero element. 

\section{Unbiased Aggregation and Convergence Analysis}\label{sec:error analysis}
In this section, 
we first design an asymptotically unbiased aggregation scheme for global update, under which, we provide the convergence analysis for the FEEL system. 
\subsection{The Design of Unbiased Aggregation}\label{eq:subsec:unbiased aggregation}
The design of unbiased aggregation should take account into two aspects of randomness in local gradients. 
One comes from stochastic compression, and the other is due to the failure in transmission within the required deadline. 
In the following, we will analyze the two parts of randomness separately, in terms of compressed gradient error and its success transmission probability, which motivate the design of unbiased aggregation. 
\subsubsection{Gradient Compression Error}
As introduced in Section \ref{subsec: compression model}, the proposed compression includes two stages, e.g., sparsification and encoding. 
We consider error only incurred by sparsification, while the encoding is error-free \cite{wangni2018gradient}.  
The following lemma provides the statistical properties for the stochastic sparsification, which performs as a surrogate for the compression error. 
\begin{lemma}\label{lemma: cmopressor properties}
	\emph{(Statistical characteristics of the compression). The stochastic compression ${\rm Comp}=\mathcal{C}\left(\mathcal{S}\left(\cdot\right)\right)$  is unbiased and the variance is given by}
	\begin{align}
		\mathbb{E}\left[\left\|{\rm Comp}\left(\mathbf{g},r\right)-\mathbf{g}\right\|^2|\mathbf{g}\right]&= \delta\left\|\mathbf{g}\right\|^2. 
	\end{align}
\end{lemma}
where $\delta$ is estimated by $\delta\approx \frac{a}{r}-1$, and $a = \frac{\left\|\mathbf{g}\right\|_1^2}{S\left\|\mathbf{g}\right\|_2^2}$. 

\begin{proof}
	See Appendix A.
\end{proof}
\subsubsection{Successful Transmission Probability}
At each communication round, each device needs to complete the global model receiving, updating, and uploading within the deadline $T_D^{(t)}$. 
The time cost of device $m$ in the  $t$-round  is given as follows. 
\begin{align}
	T_m^{(t)} = T_{{\sf B}}^{(t)} + T_{{\sf C},m}^{(t)} + T_{{\sf U},m}^{(t)} 
\end{align}
where $T_{{\sf B}}^{(t)}$, $T_{{\sf C},m}^{(t)}$, and $T_{{\sf U},m}^{(t)}$ denote the global model broadcasting time, the local computation time, and the local gradient uploading time of device $m$ at round $t$, respectively. 

Compared with computation time $T_{{\sf C},m}^{(t)}$ and uploading time $T_{{\sf U},m}^{(t)}$, the time consuming of model broadcasting is negligible due to the broadband downlink channel and the large transmitted power at the edge server \cite{zhang2022communication}.  
For simplicity, we neglect 
$T_{\sf B}^{(t)}$ in the remaining part of this paper. 
\rv{According to \cite{ren2020scheduling,tran2019federated}, 
let $\kappa$ denote the CPU cycles for a specified device to execute one batch of samples, 
the local computation time $T_{{\sf C},m}^{(t)}$ can be calculated by 
\begin{align}
    T_{{\sf C},m}^{(t)} = \frac{\kappa}{f_m}, 
\end{align}
where $f_m$ denotes the CPU frequency of device $m$. Note that we model the heterogeneity of computation by assigning devices with a different $f_m$.  
}As for the uploading time $T_{{\sf U},m}^{(t)}$, 
we compute it by using uploading bits and the data rate given in (\ref{eq: data rate definition}) as 
\begin{align}
    T_{{\sf U},m}^{(t)} = \frac{bs_m^{(t)}}{B \log_2 \left(1 + \frac{P_m|h_m^{(t)}|^2}{BN_0}\right)}, 
\end{align}
{where $b$ is the number of bits for encoding a single non-zero element, and $s_m^{(t)}$ is the number of non-zero elements after compression. 
}

As a result, we define the successful transmission by the event that the device $m$ completes gradient uploading within the deadline, and thus, the probability of successful transmission is  
\begin{align}\label{eq: definition of q}
    q_m^{(t)} \triangleq {\rm Pr}\left(T_{{\sf U},m}^{(t)}\leq T_D^{(t)}-T_{{\sf C},m}^{(t)}\right), 
\end{align}
which is determined by the joint distribution of two random variables, $h_m^{(t)}$ and $s_m^{(t)}$. 
To obtain a closed form solution, we consider the asymptotic result under the condition of $S \rightarrow \infty$, representing the large scale learning model. 
With the aforementioned channel model in Section \ref{subsec: communication model}, we have the following Lemma. 
{\begin{lemma}\label{lemma: transmission outage probability}
	\emph{(Successful Transmission Probability). 
    For device $m$ at each communication round $t$, the successful transmission probability $q_m^{(t)}$ is given by 
    }
	\begin{align}\label{eq: the computation of q}
		\lim_{S\rightarrow \infty} q_m^{(t)} = {\rm exp}\Bigg[{-\frac{BN_0}{P_m\sigma_m^2}\bigg(2^{\frac{bSr_m^{(t)}}{B\left(T_D^{(t)}-\frac{\kappa}{f_m}\right)}}-1\bigg)}\Bigg], 
	\end{align}
    \emph{where $r_m^{(t)}$ is the sparsity ratio of device $m$ at round $t$. }
\end{lemma}
}
\begin{proof}
	According to (\ref{eq: definition of q}), we have 
	\begin{align}
		\lim_{S\rightarrow \infty}q_m^{(t)} &= \lim_{S\rightarrow\infty}{\rm Pr}\Bigg[\left|h_m^{(t)}\right|^2\geq \frac{BN_0}{P_m}\bigg(2^{\frac{bs_m^{(t)}}{B\left(T_D^{(t)}-\frac{\kappa}{f_m}\right)}}-1\bigg)\Bigg] \notag\\
		&=\lim_{S\rightarrow\infty} {\rm exp}\Bigg[{-\frac{BN_0}{P_m\sigma_m^2}\bigg(2^{\frac{bs_m^{(t)}}{B\left(T_D^{(t)}-\frac{\kappa}{f_m}\right)}}-1\bigg)}\Bigg]\label{transmission probability intermideate results}
    \end{align}
	where (\ref{transmission probability intermideate results}) is because $|h_m^{(t)}|^2\sim {\rm Exp}(\lambda=\frac{1}{\sigma_m^2})$, and the desired result in Lemma \ref{lemma: transmission outage probability} is obtained by LLN. 
\end{proof}
\begin{remark}
    \rv{\emph{Lemma 2 explicitly demonstrates that the successful transmission probability is determined by various factors, namely the computation ability ($f_m$), communication capacity ($\frac{BN_0}{P_m\sigma_m^2}$), uploading data amount ($bSr_m^{(t)}$) controlled by the compression mechanism and deadline. Notably, the achievable outcomes, which include the compression mechanism and deadline, significantly impact the final value of $q_m^{(t)}$. Specifically, the choice of compression mechanism affects the actual uploading time, while the transmission deadline sets the permissible time window. It is important to note that optimizing one aspect without considering the other would not lead to optimal control over transmission outage. This highlights the strong coupling between compression and deadline and emphasizes the necessity of a joint optimization strategy.}}
\end{remark}
\subsubsection{Aggregation Design}
To propose the unbiased aggregation scheme as in \cite{zhang2022communication,ren2020scheduling}, we write the received local gradient as the following form 
\begin{align}\label{eq:received gradient expression}
	\widehat{\mathbf{g}}_m^{(t)} = Y_m^{(t)}\left(\mathbf{g}_m^{(t)}+\mathbf{e}_m^{(t)}\right), 
\end{align}
where $\mathbf{e}_m^{(t)}$ denotes the {gradient compression error} with mean and variance given in Lemma \ref{lemma: cmopressor properties}.  
$Y_m^{(t)}$ is an $0$-$1$ indicator of successful transmission, and $Y_m^{(t)}=0$ means {transmission outage}  occurs.  
It follows
Bernoulli distribution with the mean equal to 
successful transmission probability $q_m^{(t)}$ as provided in Lemma \ref{lemma: transmission outage probability}. 

With the form of received local gradients in (\ref{eq:received gradient expression}), we propose the asymptotically unbiased aggregation as 
\begin{align}\label{eq:aggregation scheme}
	\widehat{\mathbf{g}}^{(t)} = \sum_{m=1}^M \frac{d_m}{q_m^{(t)}d}\widehat{\mathbf{g}}_m^{(t)}.
\end{align}
The proposed design is obtained by the unbiased compression, and by LLN to decouple the dependence of compression and transmission outage as\footnote{{Note that, although $\widehat{\mathbf{g}}_m^{(t)}$ is unbiased asymptotically (i.e., $S\rightarrow +\infty$), it generally matches well since $S$ in real FEEL scenario is usually large, e.g., $S=2.56\times 10^{7}$ in ResNet50 \cite{he2016deep}.}} 
\begin{align}	
    \lim_{S\rightarrow\infty}\mathbb{E}\left(\widehat{\mathbf{g}}^{(t)}\right)&=\lim_{S\rightarrow\infty}\sum_{m=1}^M\mathbb{E}_{Y_m^{(t)}}\mathbb{E}_{\mathbf{e}_m^{(t)}}\Big[\frac{d_mY_m^{(t)}}{dq_m^{(t)}}\left(\mathbf{g}_m^{(t)}+\mathbf{e}_m^{(t)}\right)\Big]\notag\\
    =&~~\lim_{S\rightarrow\infty}\sum_{m=1}^M\frac{d_m}{dq_m^{(t)}}\mathbb{E}_{Y_m^{(t)}}\left[Y_m^{(t)}\right]\mathbb{E}_{\mathbf{e}_m^{(t)}}\left[\mathbf{g}_m^{(t)}+\mathbf{e}_m^{(t)}\right]\label{eq: intermideate results of unbiased scheme}\notag\\
    &=\sum_{m=1}^M \frac{d_m}{d}\mathbf{g}_m^{(t)}\notag=\mathbf{g}^{(t)}\notag.
\end{align}
\vspace{-5mm}
\subsection{Convergence Analysis}

To begin with, four standard assumptions on the loss function are made. 
\begin{assumption}
	The loss function $L\left(\cdot\right)$ is $\ell$-smooth, i.e., $\forall$ $\mathbf{u}$ and $\mathbf{v}$, 
	$L\left(\mathbf{u}\right)\leq L\left(\mathbf{v}\right) + \nabla L\left(\mathbf{v}\right)^T \left(\mathbf{u} - \mathbf{v}\right)+ \frac{\ell}{2}\left\|\mathbf{u}-\mathbf{v}\right\|^2$.  
\end{assumption}
\begin{assumption}
	The loss function $L\left(\cdot\right)$ is $\mu$-strongly-convex, i.e., $\forall~\mathbf{u}$ and $\mathbf{v}$, 
	$L\left(\mathbf{u}\right)\geq L\left(\mathbf{v}\right) + \nabla L\left(\mathbf{v}\right)^T \left(\mathbf{u} - \mathbf{v}\right)+ \frac{\mu}{2}\left\|\mathbf{u}-\mathbf{v}\right\|^2$.   
\end{assumption}
\begin{assumption}
	The stochastic gradients obtained through SGD algorithm are unbiased and the variance is bounded, i.e., $\forall m$ and $\mathbf{w}$, $\mathbb{E}[\mathbf{g}_m^{(t)}]=\nabla L_m\left(\mathbf{w},\mathcal{D}_m\right)$, and $\mathbb{E}[\|\mathbf{g}_m^{(t)}-\nabla L_m\left(\mathbf{w},\mathcal{D}_m\right)\|^2]\leq \sigma^2$. 
\end{assumption}
\begin{assumption}\label{assumption: gradient bound}
    {The expected norm of the stochastic gradient $\mathbf{g}_m^{(t)}$ is upper bounded by $\mathbb{E}\|\mathbf{g}_m^{(t)}\|_2^2\leq G, \forall m,t.$}
\end{assumption}
Based on the above assumptions, 
we derive the convergence result for future $l$ iterations at each round $t$ as below. 
\begin{proposition}\label{prop: convergence analysis}
\emph{ For FEEL system with sparse coefficient $\{r_m^{(n)},n\geq t, m=1,...,M\}$, the deadline $\{T_D^{(n)},n\geq t\}$, the fractional decay learning rate $\eta^{(t)}=\frac{\chi}{t+\nu}\leq \frac{1}{2L}$ and the current global model $\mathbf{w}^{(t)}$. At each round $t$, by updating $l$ more rounds, the average optimality gap is upper bound by}
\begin{align}
	\mathbb{E}&\left(L\left(\mathbf{w}^{(t+l)}\right)-L\left(\mathbf{w}^*\right)\right) \notag
	\\
    &\leq \frac{1}{t+\nu+l}\bigg(\underbrace{{\frac{\ell(t+\nu)}{\mu}\left(L\left(\mathbf{w}^{(t)}\right)-L\left(\mathbf{w}^*\right)\right)}}_{(a)}\notag\\&+\underbrace{A\sum_{m=1}^M\frac{d_m^2}{d^2}\sigma^2}_{(b)}+\underbrace{AG\sum_{m=1}^M\big(\max_{n=t,...,t+l}\frac{a_m^{(n)}}{r_m^{(n)}q_m^{(n)}}-1\big)}_{(c)}\bigg), 
\end{align}
\emph{where $\chi$ and $\nu$ are hyperparameters to control the learning rate, $A = \frac{\ell \chi^2 }{ \left(3\mu \chi -2\right)}$, 
}
$a_m^{(n)}=\frac{\left\|\mathbf{g}_m^{(n)}\right\|_1^2}{S\left\|\mathbf{g}_m^{(n)}\right\|_2^2}$. 
\end{proposition}
\begin{proof}
	See Appendix B. 
\end{proof}
\begin{remark}\label{remark: convergence analysis}
	\emph{(Convergence analysis). Proposition 1 describes the convergence rate of FEEL based on gradient compression in digital wireless communication scenarios. 
	The upper bound is composed of three terms. 
	The terms $(a)$ is an inherent term, which is determined by the optimality gap between the current loss and global minimum. 
	The term $(b)$ is caused by the gradient variance introduced by the SGD algorithm. 
	\rv{The term $(c)$ is our main concern. 
	This term quantifies the impact of compression and transmission outage on convergence. 
    Specifically, a small $r_m^{(n)}$ implies a large compression error as characterized in Lemma \ref{lemma: cmopressor properties}, and a small $q_m^{(n)}$ means a large outage probability. 
    In essence, both the compression error and outage rate affect the expected convergence performance. 
    They contribute to widening the optimality gap between the converged model and the optimal model by 
    increasing the mixed term $\frac{1}{r_m^{(n)}q_m^{(n)}}$.
    }
	}
\end{remark}
\section{Joint Compression and Deadline Optimization}\label{sec:problem formulation}
In this section, 
we leverage Proposition \ref{prop: convergence analysis} to address the problem of minimizing the training time over the compression ratio and deadline setting. 
The problem is solved in an alternating manner.   
We further conclude by presenting the proposed \emph{joint compression and deadline optimization} (JCDO) algorithm. 
\subsection{Problem Formulation}
To begin with, we consider the metric $N_{\epsilon}^{(t)}$ as the number of rounds to achieve $\epsilon$-\emph{accuracy} when sitting at round-$t$, as follows. 
\begin{align}
    L\left(\mathbf{w}^{(t+N_{\epsilon}^{(t)})}\right)-L\left(\mathbf{w}^{(*)}\right)\leq \epsilon. 
\end{align}
\rv{For the FEEL framework depicted in Fig. \ref{Fig:1}, the primary objective is to minimize the training time required for the global model to converge. 
We in this paper try to achieve this by performing data compression at the device side and setting a deadline at the server side. 
Given the compression scheme and deadline scheme in Section \ref{sec: system model}, the compression ratio and the value of deadline need to be determined. 
The overall programming problem is given by
}\begin{align}
    \mathscr{P}_1:&\min\limits _{\{r_m^{(n)}\}_{n=1,m=1}^{N,M},\{T_{D}^{(n)}\}_{n=1}^N}&&{\displaystyle\sum_{n=1}^{N} T_D^{(n)}},\\
    &\qquad\qquad~{\rm s.t.} &&\mathbb{E}\big\{L\big(\mathbf{w}^{(N)}\big)-L\left(\mathbf{w}^*\right)\big\}\leq \epsilon, \notag\\
    & &&T_{D}^{(n)} \geq T_{{\sf C},m}^{(n)}, \forall n, m, \notag\\
    & && 0< r_m^{(n)} \leq 1, \forall n, m, \notag
\end{align}
    where   
    $T_{D}^{(n)}$ denotes the deadline set by the edge server at round $n$. 
    As we focus on the uploading process, we assume that each device should at least have a certain time to upload its local update, that is, $T_D^{(n)}$ should be larger than the maximum of the local computation time. 
    \rv{$N$ denotes the specific round at which the global model converges. 
    We note that the value of $N$ is intrinsically linked to the optimality gap presented in Proposition \ref{prop: convergence analysis}, and a higher value of $N$ means an increased optimality gap. 
    }
    
    $\mathscr{P}_1$ is not tractable. Because the convergence manner cannot be well captured at the initial round, and the non-causal information  including $\{\{a_m^{(n)}\},\{h_m^{(n)}\},n>1\}$  is required. 
    We thus consider optimizing $\{r_m^{(n)}\}_{n=1,m=1}^{N,M}$ and $\{T_{D}^{(n)}\}_{n=1}^N$ in an online manner, 
    that is,  the compression ratio and deadline are optimized at the beginning of each round. 
    Moreover, 
    We adopt the idea of dynamic programming to make the solving process causal:  
    When sitting at round $t$ for determining the compression ratio and deadline, we assume the same strategy is adopted in the subsequent rounds, that is, $r_m^{(n)}=r_m^{(t)}$ and $T_D^{(n)}=T_D^{(t)}, \forall n\geq t$. 
    Then the problem instance for round $t$ is given by
\begin{align}
    \mathscr{P}_2:&\min\limits _{\{r_m^{(t)}\}_{m=1}^M,{T_{D}^{(t)}}}&& {\displaystyle\sum_{n=t}^{t+N_{\epsilon}^{(t)}}T_D^{(n)}},\hfill\\
    &~~~~~~~{\rm s.t.} &&\mathbb{E}\big\{L\big(\mathbf{w}^{(t+N_{\epsilon}^{(t)})}\big)-L\left(\mathbf{w}^*\right)\big\}\leq \epsilon,\hfill \notag\\
    & &&T_{D}^{(t)} \geq T_{{\sf C},m}^{(t)},T_{D}^{(n)}=T_D^{(t)}, \forall n\geq t, m, \hfill\notag\\
    & && 0< r_m^{(t)} \leq 1,r_m^{(n)}=r_m^{(t)}, \forall n\geq t, m.  \hfill \notag
\end{align}
Furthermore,  the main challenge is to derive $N_{\epsilon}^{(t)}$. The additional assumption is made as follows. 
\begin{assumption}\label{assumption: gradient norm bounded}
    {  Assume the FEEL system has updated $t$ rounds, $a_m^{(n)}$ in the following rounds is upper bounded, i.e., $a_m^{(n)}\leq \alpha_m^{(t)}, n=t,t+1,...,t+N_{\epsilon}^{(t)}.$}
\end{assumption}
As the same strategy is adopted in the remaining rounds, 
combined with Assumption \ref{assumption: gradient norm bounded}, the maximize operation in term $(c)$ of Proposition \ref{prop: convergence analysis} can be safely removed. 
We have the following lemma.
\begin{lemma}\label{lemma: remaining communication rounds}
    \emph{For FEEL system with sparse coefficient $\{r_m^{(n)}=r_m^{(t)},n\geq t,m=1,...,M\}$, the deadline $\{T_D^{(n)}=T_D^{(t)}, n\geq t\}$, and the fractional decay learning rate $\eta^{(t)}=\frac{\chi}{t+\nu}\leq \frac{1}{2L}$. 
    To meet the convergence requirement $\epsilon$, the number of remaining communication rounds should be carried out is upper bounded by}  
    \begin{align}
        N_{\epsilon}^{(t)}&\leq \frac{\ell (t+\nu)}{\mu \epsilon}\left(L(\mathbf{w}^{(t)})-L(\mathbf{w}^*)\right)-t-\nu \notag\\&~~~+ \frac{A}{\epsilon}\sum_{m=1}^M\frac{d_m^2}{d^2}\sigma^2 + \frac{AG}{\epsilon}\sum_{m=1}^M\frac{d_m^2}{d^2}\bigg(\frac{\alpha_m^{(t)}}{r_m^{(t)}q_m^{(t)}}-1\bigg).
    \end{align}
\end{lemma} 
\begin{proof}
    See Appendix C. 
\end{proof}

We estimate $N_{\epsilon}^{(t)}$ by its upper bound in Lemma \ref{lemma: remaining communication rounds}. 
Then the problem $\mathscr{P}_2$ is reduced to 
\begin{align}
    \begin{matrix}
        \mathscr{P}_3:&\min\limits _{\left\{r_m^{(t)}\right\},T_{D}^{(t)}}& T_{D}^{(t)}\left(B_t + \sum_{m=1}^M\frac{d_m^2}{d^2}\left(\frac{\alpha_m^{(t)}}{r_m^{(t)}q_m^{(t)}}-1\right)\right),\hfill \\
        &{\rm s.t.} 
         &T_{D}^{(t)} \geq T_{{\sf C},m}^{(t)}, \forall m, \hfill\\
        & & 0< r_m^{(t)} \leq 1, \forall m, \hfill
      \end{matrix}
\end{align}
where $B_t = \frac{(t+\nu)(3\mu\chi-2)}{\mu\chi^2G} \left(L\left(\mathbf{w}^{(t)}\right)-L\left(\mathbf{w}^*\right)-\frac{\mu}{\ell}\epsilon\right)+\sum_{m=1}^M\frac{d_m^2}{d^2}\frac{\sigma^2}{G}$, denoting the current training state of the global model $\mathbf{w}^{(t)}$.  
\subsection{Compression Ratio and Deadline Optimization}
$\mathscr{P}_3$ is a non-convex multivariable problem, we thus consider solving it in an alternating manner, i.e., optimizing one variable by considering the other to be given and fixed in each iteration.
\subsubsection{Compression Ratio Optimization}\label{subsec: sparse coefficient optimization}
In this subsection, we consider optimizing the compression ratio, that is, the sparse coefficient $r_m^{(t)}$ with a fixed $T_{D}^{(t)}$, 
the subproblem is given below. 
\begin{align}
	\mathscr{P}_4: \min_{r_m^{(t)}}~&
	\frac{1}{r_m^{(t)}q_m^{(t)}}\label{eq: comprehensive error minimization problem},\\
	{\rm s.t.}~& 0< r_m^{(t)}\leq 1. \notag
\end{align}
\begin{proposition}\label{proposition: optimal solution of sparse coefficient}
	\emph{Given $T_D^{(t)}$, for FEEL system with stochastic compression and transmission outage, the optimal sparse coefficient for error minimization is given by} 
	\begin{align}\label{eq: optimal solution of sparse ratio}
		r_m^{(t)*}=\min\left\{\frac{B(T_{D}^{(t)} -\frac{\kappa}{f_m})}{bS}h^{-1}\left(\frac{P_m\sigma_m^2}{BN_0\ln 2}\right),1\right\},
	\end{align}
	\emph{where $h\left(x\right)=x2^x$, $h^{-1}\left(x\right)$ denotes the inverse function of $h(x)$.}
\end{proposition}
\begin{proof}
	See Appendix D.
\end{proof}
\begin{remark}
	\emph{ (Trade-off between compression error and transmission outage). As shown in $\mathscr{P}_4$, there is a trade-off between compression error and transmission outage when setting $r_m^{(t)}$. A larger $r_m^{(t)}$ naturally leads to a smaller compression error, while due to the increase of uploaded bits, the successful transmission probability $q_m^{(t)}$ decreases. 
	Therefore, the mixed term (i.e., $\frac{1}{r_m^{(t)}q_m^{(t)}}$) rather than individual compression error or transmission outage should be minimized when considering the $r_m^{(t)}$ setting.
	As observed in Proposition 2, $r_m^{(t)*}$ is governed by two terms, i.e., the one related to specific transmission task $\big(\frac{T_D^{(t)}-\frac{\kappa}{f_m}}{b}\big)$, the one related to transmission rate $\big(\frac{P_m\sigma_m^2}{BN_0}\big)$. 
 $r_m^{(t)*}$ monotonically increases with both terms.
 This indicates that when the transmission conditions is good (i.e., enough deadline or larger transmission rate), the optimal sparser will reserve more elements so as to reduce compression error. On the other hand, when conditions become worse, the sparser will turns its priority to suppress transmission outage by dropping more elements.}
\end{remark}
\subsubsection{Deadline Optimization}
\begin{figure}[h]
    \label{alg:LSB}
    \begin{algorithm}[H]
      \caption{The proposed JCDO algorithm}\label{algo:sampling with relpacement}
      \begin{algorithmic}[1]
        \vspace{3mm}
          \renewcommand{\algorithmicrequire}{ \textbf{Initialize}}
          \REQUIRE the global model $\mathbf{w}^{(1)}$ by the server. 
          \FOR {$t=1, 2, ..., N$}
          \STATE Initialize the received update set $\mathcal{K}^{(t)}=\emptyset$. 
          \STATE Calculate $B_t$, $\left\{P_m\right\}_{m=1}^M$, $\left\{\sigma_m^2\right\}_{m=1}^M$, $\{\alpha_m^{(t)}\}_{m=1}^M$, $G$ based on the $\mathbf{w}^{(t)}$ and the current channel state.
          \STATE $\{r_m^{(t)}\}_{m=1}^M$, $T_{D}^{(t)} =$ TransmissionPlan$(B_t, \left\{P_m\right\}_{m=1}^M,$\\$\qquad\qquad\qquad\qquad\qquad~~~~~~~\left\{\sigma_m^2\right\}_{m=1}^M,\{\alpha_m^{(t)}\}_{m=1}^M, G)$.
          \STATE BS broadcasts $\mathbf{w}^{(t)}$ to devices, and starts the countdown. // {\texttt{Model Broadcasting}}
          \FOR {$m=1,2,...,M$ {\bf in parallel}}
          \STATE Set the local model as $\mathbf{w}^{(t)}$.
          \STATE Obtain the local gradient $\mathbf{g}_m^{(t)}$ through the SGD algorithm.~// {\texttt{Local Update}}
          \STATE Send the compressed verion ${\rm Comp}\big(\mathbf{g}_m^{(t)}, r_m^{(t)}\big)$ to the server.~// {\texttt{Update Uploading}}
          \IF{$T_{{\sf C},m}^{(t)}+T_{{\sf U},m}^{(t)}\leq T_D^{(t)}$}
          \STATE $\mathcal{K}^{(t)} = \mathcal{K}^{(t)} \cup \left\{m\right\}$. 
          \ENDIF
          \ENDFOR
          \STATE $\mathbf{w}^{(t+1)}=\mathbf{w}^{(t)}-\eta^{(t)}\sum_{m \in \mathcal{K}^{(t)}} \frac{d_m}{dq_m^{(t)}}{\rm Comp}\big(\mathbf{g}_m^{(t)},r_m^{(t)}\big)$. // {\texttt{Aggregation}}
          \ENDFOR
      \end{algorithmic}
      \begin{algorithmic}[1]
      \renewcommand{\algorithmicrequire}{ \textbf{Function}}
          \REQUIRE TransmissionPlan$\big(B_t, \left\{P_m\right\}_{m=1}^M,\left\{\sigma_m^2\right\}_{m=1}^M,$\\\qquad\qquad\qquad\qquad\qquad\qquad\qquad\qquad\qquad$\{\alpha_m^{(t)}\}_{m=1}^M, G\big)$
          \STATE {\bf Initialize} $T_{D}^{(t)}=T_{D}^{(t-1)},T_{t,{\rm before}}^D=0$.
          \WHILE{$|T_{D}^{(t)} - T_{t,{\rm before}}^D|\geq \alpha$}
          \STATE $T_{t,{\rm before}}^D=T_{D}^{(t)}$.
          \FOR {$m=1,2,...,M$ {\bf in parallel}}
          \STATE Updates $r_m^{(t)}$ using (\ref{eq: optimal solution of sparse ratio}) and $T_D^{(t)}$.
          \ENDFOR
          \STATE Updates $T_{D}^{(t)}$ by solving $\mathscr{P}_5$ with $\{r_m^{(t)}\}_{m=1}^M$.
          \ENDWHILE 
          \STATE {\bf return} $\{r_m^{(t)}\}_{m=1}^M$, $T_D^{(t)}$.
      \end{algorithmic}
    \end{algorithm}
  \end{figure}
In this subsection, we consider optimizing the deadline $T_{D}^{(t)}$ with a fixed $r_m^{(t)}$, similar to Section \ref{subsec: sparse coefficient optimization}, the subproblem is given as follows. 
\begin{align}
	\mathscr{P}_5:\min_{T_{D}^{(t)}}~~&B_t T_{D}^{(t)}\notag\\~~~+ T_{D}^{(t)} &\sum_{m=1}^M \frac{d_m^2}{d^2} \frac{\alpha_m^{(t)}}{r_m^{(t)}} {\rm exp}\Bigg[{{\frac{BN_0}{P\sigma_m^2}}\bigg({2^{\frac{C_m^{(t)}}{T_{D}^{(t)}-T_{{\sf C},m}^{(t)}}}-1}\bigg)}\Bigg], \\
	{\rm s.t.}~~&T_{D}^{(t)} \geq  T_{{\sf C},m}^{(t)}.\notag  
\end{align}
where $C_m^{(t)}=\frac{bSr_m^{(t)}}{B}$.  
\begin{lemma}\label{convex lemma}
	\emph{$\mathscr{P}_5$ is convex.}
\end{lemma}
\begin{proof}
	See Appendix E.
\end{proof}

With Lemma \ref{convex lemma}, the optimal $T_{D}^{(t)}$ could be obtained efficiently through bisection algorithm. 
\begin{remark}
	\emph{(Adaptive adjustment of deadline).
	The edge server needs to set the deadline according to the current global model and the channel state. 
	Intuitively, reducing the deadline can reduce the training latency of each round. 
	However, for the device side, it will increase the compression error w.r.t. the original local update, as devices have to increase the data compression degree to avoid to increase the outage probability. 
	According to Remark \ref{remark: convergence analysis}, a larger compression error requires more rounds, thus a smaller deadline does not necessarily lead to less training time. 
	On the other hand, the optimal deadline also depends on the current training state through the related term $B_t$. 
	$B_t$ changes over communication rounds, suggesting the continuous dynamic adjustment of the receiving window, which is more conducive to reducing the remaining training time. 
	}
\end{remark}

The overall algorithm is presented in Algorithm \ref{algo:sampling with relpacement}\footnote{Note that, at round $t$, we estimate $\{\alpha_m^{(t)}\}_{m=1}^M$, $G$ by $G\approx \max_{n=1,...,t,m=1,...,M}\|\mathbf{g}_m^{(n)}\|_2^2$, and $\alpha_m^{(t)}\approx\max_{n=1,...,t}\frac{\|\mathbf{g}_m^{(n)}\|_1^2}{S\|\mathbf{g}_m^{(n)}\|_2^2}$.}. 
\begin{table*}[h]
    \scriptsize
    \caption{Datasets and DNNs adopted in this paper.}\label{table: dataset and model}
    \vspace{0.1in}
    \resizebox{1\textwidth}{18mm}
    {
    \centering
    \begin{tabular}{|c|c|c|c|c|c|c|c|c|}
    \hline
    \hline
    \textbf{Dataset} & \textbf{Task} & \multicolumn{2}{c|}{\textbf{Characterization}} & \textbf{Training Set} & \textbf{Testing Set} & \textbf{Data Format} & \textbf{Label Format}
    \\ \hline
     \multirow{2}{*}{\begin{tabular}[|c|]{@{}c@{}} FEMNIST \cite{caldas2018leaf} \end{tabular}}
     &
     \multirow{2}{*}{\begin{tabular}[|c|]{@{}c@{}} Image\\classification \end{tabular}}
      &
      \multicolumn{2}{c|}{
        \multirow{2}{*}{\begin{tabular}[c|]{@{}c@{}} handdwritten images, consists of 10 numbers,\\ 26 lower case letters, 26 upper case letters  \end{tabular}}
        }
    &	\multirow{2}{*}{\begin{tabular}[|c|]{@{}c@{}} $220313$ images,\\ $100$ users \end{tabular}}
       &
            \multirow{2}{*}{\begin{tabular}[|c|]{@{}c@{}} $25025$ images, $404$ for \\each class in average\end{tabular}}
       &
     \multirow{2}{*}{\begin{tabular}[|c|]{@{}c@{}} $28\times28$ grayscale image\end{tabular}}
          &
     \multirow{2}{*}{\begin{tabular}[c]{@{}c@{}} Class index: $0$-$61$   \end{tabular}}
    \\
    & & \multicolumn{2}{c|}{} & & & &
    \\
    \hline
     \multirow{2}{*}{\begin{tabular}[|c|]{@{}c@{}} CIFAR10 \cite{krizhevsky2009learning} \end{tabular}}
     &
     \multirow{2}{*}{\begin{tabular}[|c|]{@{}c@{}} Image\\classification \end{tabular}}
      &
      \multicolumn{2}{c|}{
        \multirow{2}{*}{\begin{tabular}[c|]{@{}c@{}}  a widely-adopted tiny image dataset, \\contains some common classes \end{tabular}}
        }
       &
            \multirow{2}{*}{\begin{tabular}[|c|]{@{}c@{}} $50000$ images, \\each class has $5000$\end{tabular}}
       &
            \multirow{2}{*}{\begin{tabular}[|c|]{@{}c@{}} $10000$ images, \\each class has $1000$ \end{tabular}}
       &
     \multirow{2}{*}{\begin{tabular}[|c|]{@{}c@{}} $32\times32$ colour image \end{tabular}}
          &
     \multirow{2}{*}{\begin{tabular}[c]{@{}c@{}} Class index: $0$-$9$  \end{tabular}}
    \\
    & & \multicolumn{2}{c|}{} & & & &
    \\
    \hline
    \multirow{2}{*}{\begin{tabular}[|c|]{@{}c@{}} KITTI \cite{Geiger2012CVPR} \end{tabular}}
    &
    \multirow{2}{*}{\begin{tabular}[|c|]{@{}c@{}} Object\\detection \end{tabular}}
     &
     \multicolumn{2}{c|}{
       \multirow{2}{*}{\begin{tabular}[c|]{@{}c@{}} 2D image dataset for vehicle \\detection in automatic driving \end{tabular}}
       }
      &
           \multirow{2}{*}{\begin{tabular}[|c|]{@{}c@{}} $7000$ images\\$38094$ labeled objects\end{tabular}}
      &
           \multirow{2}{*}{\begin{tabular}[|c|]{@{}c@{}} $480$ images\\$2476$ labeled objects \end{tabular}}
      &
    \multirow{2}{*}{\begin{tabular}[|c|]{@{}c@{}} $1224\times 370$ colour image \end{tabular}}
         &
    \multirow{2}{*}{\begin{tabular}[c]{@{}c@{}} Object list in KITTI format\\$5$ elements for each object \end{tabular}}
   \\
   & & \multicolumn{2}{c|}{} & & & &
   \\
   \hline
    \textbf{DNN} & \textbf{Task} &
    \multicolumn{3}{c|}{\textbf{Architecture }}
    & \textbf{Loss Function} & \textbf{Learning Rate} & \textbf{Performance Metric}
    \\
    \hline
     \multirow{2}{*}{\begin{tabular}[|c|]{@{}c@{}} Logistic \end{tabular}}
     &
     \multirow{2}{*}{\begin{tabular}[|c|]{@{}c@{}} Image\\classification \end{tabular}}
       &\multicolumn{3}{c|}{
       \multirow{2}{*}{\begin{tabular}[c|]{@{}c@{}} $1$ fully-connected layer, $1$ soft-max output layer,\\ $48670$ trainable parameters \end{tabular}}
       }
       &
            \multirow{2}{*}{\begin{tabular}[|c|]{@{}c@{}} Cross Entropy \end{tabular}}
       &
     \multirow{2}{*}{\begin{tabular}[|c|]{@{}c@{}} $\chi=30,\nu=100$, SGD \end{tabular}}
          &
     \multirow{2}{*}{\begin{tabular}[c]{@{}c@{}} Classification accuracy \end{tabular}}
    \\
    & & \multicolumn{3}{c|}{} & & &
    \\
    \cline{1-8}
     \multirow{2}{*}{\begin{tabular}[|c|]{@{}c@{}} VGG \cite{simonyan2014very} \end{tabular}}
     &
     \multirow{2}{*}{\begin{tabular}[|c|]{@{}c@{}} Image\\classification \end{tabular}}
       &\multicolumn{3}{c|}{
       \multirow{2}{*}{\begin{tabular}[c|]{@{}c@{}} $8$ VGG blocks, $5$ pooling layers, classifier with dropout, \\ $9756426$ trainable parameters \end{tabular}}
       }
       &
       \multirow{2}{*}{\begin{tabular}[|c|]{@{}c@{}} Cross Entropy \end{tabular}}
       &
     \multirow{2}{*}{\begin{tabular}[|c|]{@{}c@{}} $\chi=10,\nu=100$, SGD  \end{tabular}}
          &
     \multirow{2}{*}{\begin{tabular}[c]{@{}c@{}} Classification accuracy \end{tabular}}
    \\
    & & \multicolumn{3}{c|}{} & & &
    \\
    \hline
    \cline{1-8}
    \multirow{2}{*}{\begin{tabular}[|c|]{@{}c@{}} YOLOv5  \end{tabular}}
    &
    \multirow{2}{*}{\begin{tabular}[|c|]{@{}c@{}} Object\\detection \end{tabular}}
      &\multicolumn{3}{c|}{
      \multirow{2}{*}{\begin{tabular}[c|]{@{}c@{}} Backbone: CSPDarknet, Neck: PANet, Head: YOLO layer\\ $7041205$ trainable parameters \end{tabular}}
      }
      &
      \multirow{2}{*}{\begin{tabular}[|c|]{@{}c@{}} BCEWithLogitsLoss, CIoU loss \end{tabular}}
      &
    \multirow{2}{*}{\begin{tabular}[|c|]{@{}c@{}} $\chi=3,\nu=100$, Adam  \end{tabular}}
         &
    \multirow{2}{*}{\begin{tabular}[c]{@{}c@{}} mean Average precision (mAP) \\at $\rm{IoU}=0.5$ \end{tabular}}
   \\
   & & \multicolumn{3}{c|}{} & & &
   \\
   \hline
   
    \end{tabular}
    }
    \label{Table.related_work}
    \end{table*}
\begin{remark}\label{remark: time complexity analysis}
	\emph{(Computational complexity analysis). 
    The execution of the proposed JCDO algorithm in each round includes 
    the alternating update of the compression ratios $\{r_{m}^{(t)}\}_{m=1}^M$ and the deadline $T_D^{(t)}$. 
    For updating $\{r_{m}^{(t)}\}_{m=1}^M$, 
    as the semi closed-form solution is given in (\ref{eq: optimal solution of sparse ratio}), the complexity for this part is $\mathcal{O}(M\log_2(\mathscr{R}_1))$, where $\mathscr{R}_1$ denotes the searching space for obtaining the value of $h^{-1}(x)$;  
    For the deadline optimization, the required complexity is $\mathcal{O}(\log_2(\mathscr{R}_2))$, where $\mathscr{R}_2$ is the search space of $T_D^{(t)}$. 
    Therefore, the overall complexity for transmission plan is $\mathcal{O}(J(M\log_2(\mathscr{R}_1))+\log_2(\mathscr{R}_2))$, where $J$ is the number of iterations. 
    }
\end{remark}
\begin{table}[h]
    \caption{{System settings}}\label{table: communication settings}
    \centering
    \begin{tabular}{|c|c|c|c|}
    \hline
    Parameters & FEMNIST & CIFAR10 & KITTI \\ \hline
    $M$          & $100$     & $10$      & $10$    \\ \hline
    $B$          & 1MHz    &     1MHz    &    1MHz   \\ \hline
    $N_0$          & $-174$ dBm/Hz    &     $-174$ dBm/Hz     &    $-174$ dBm/Hz    \\ \hline
    $P_m$         & $8$dBm   &   $18$dBm      &   $18$dBm    \\ \hline
    $b$          & $16$      & $32$      & $32$    \\ \hline
    $\kappa$        &   $5\times10^4$      &   $5\times 10^6$      &   $3\times 10^{8}$    \\ \hline
    $f_m$       &    $\mathcal{U}[0.1,1]$ GHz      &    $\mathcal{U}[0.1,1]$ GHz     &   $\mathcal{U}[0.1,1]$ GHz    \\ \hline
    \end{tabular}
    \end{table}
\section{Experimental Results}
In this section, we aim to demonstrate the effectiveness of our proposed JCDO algorithm. 
To this end, we consider the following benchmark schemes: 
1) FedSGD \cite{mcmahan2017communication}, which uploads the model updates in full precision and waits for the arrive of the slowest device update before starting aggregation process;  
2) FedTOE \cite{{wang2021quantized}}, which aims to adjust the compression ratio to keep the devices having the same outage probability. 
In addition, to verify the effectiveness of compressor and deadline optimization separately, two algorithms extended from JCDO algorithm, namely \emph{compression-optimization} (CO) with a fixed deadline, and \emph{deadline optimization} (DO) with a fixed sparse coefficient are considered as well.

\rv{We consider two federated learning tasks, namely standard task and practical task. 
For standard tasks, we consider image classification on FEMNIST dataset and CIFAR10 dataset, using Logistic regression model and VGG model, respectively. 
}
\rv{Additionally, we aim to test on real-world advanced FEEL applications, which often involve more complex tasks than classification. Therefore to bridge this gap, we've included object detection tasks in autonomous driving applications.}
Specifically,  we consider training the well known YOLOV5 model on a realistic autonomous driving dataset, KITTI. 
The training details and the characterization of datasets are presented in Table \ref{table: dataset and model}. 
{The number of device, communication settings,  and computation parameters for the aforementioned tasks are concluded in Table \ref{table: communication settings}.} 

\rv{Notably, in practical FEEL scenarios, there are many aspects of heterogeneity among devices. 
Given this, similar to \cite{jiang2022adaptive}, we in this paper consider the heterogeneity issue from three aspects, namely dataset heterogeneity, computation heterogeneity, and communication heterogeneity:  
1) Dataset heterogeneity: We use a label aware approach to partition the CIFAR-10 dataset. 
Specifically, the whole CIFAR10 dataset is first sorted by the labels, then is divided into $40$ shards of size $1250$. 
Each device is assigned with 4 shards. 
The FEMNIST dataset is a benchmark dataset for federated learning, it is generated from multiple users, thus does not need to be processed manually. 
As for the KITTI dataset, we partitioned the dataset in a random manner, that is, each device is randomly assigned $700$ training samples from the whole dataset without replacement. 
2) Computation heterogeneity: as shown in Table \ref{table: communication settings}, the computational capabilities denoted by $f_m$ varies among devices, and follows a uniform distribution ranging from $0.1$GHz to $1$GHz. 
3) Communication heterogeneity: The distance $\omega$ (in km) between the edge server and devices is uniformly distributed between $0.01$ and $0.5$, and the corresponding path loss (i.e., $\sigma_m^2$) defined in Section \ref{subsec: communication model} is $128.1+37.6\log_{10}\left(\omega\right)$ in dB. 
}

All the experiments are implemented by PyTorch 1.8.0 and Python 3.8 on a Linux server with 2 NVIDIA RTX 3090 GPUs.


\subsection{\rv{FEEL for {Standard Tasks}: Classification on FEMNIST and CIFAR-10 Datasets}}
\begin{figure*}[h]
    \centering
    \subfigure[Logistic on FEMNIST]{
        \label{Fig:compression convergence curve:a} 
        \includegraphics[width=0.235\linewidth]{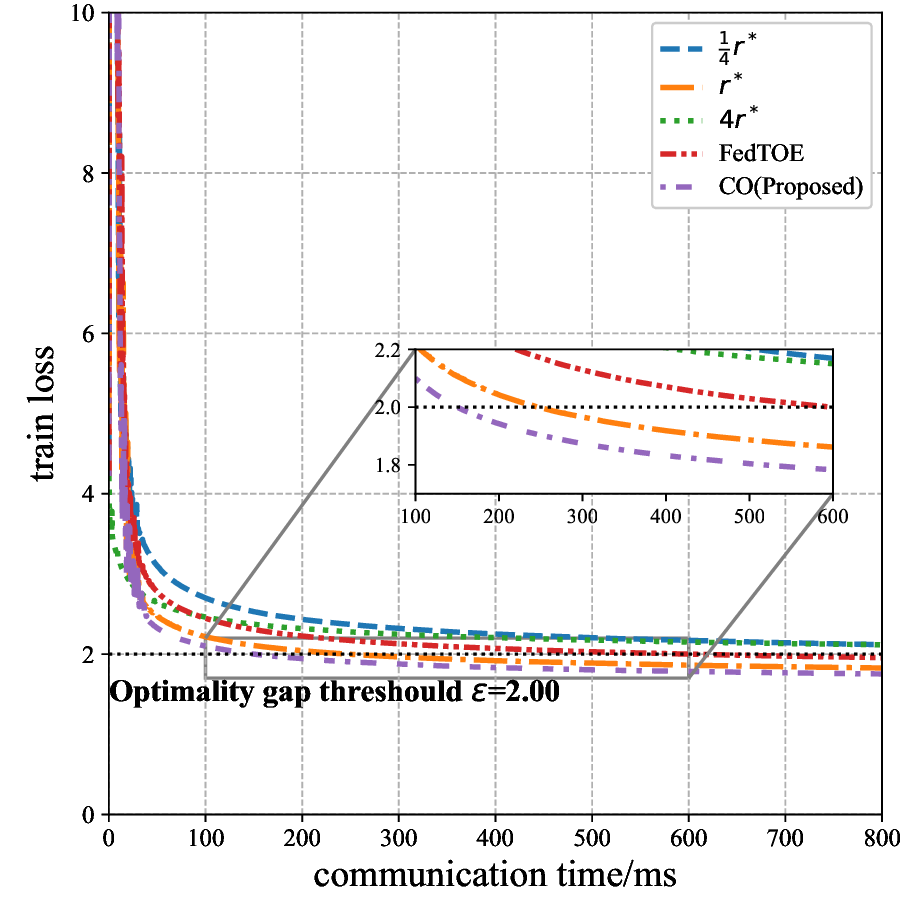}}
    \subfigure[Logistic on FEMNIST]{
        \label{Fig:compression convergence curve:b} 
        \includegraphics[width=0.235\linewidth]{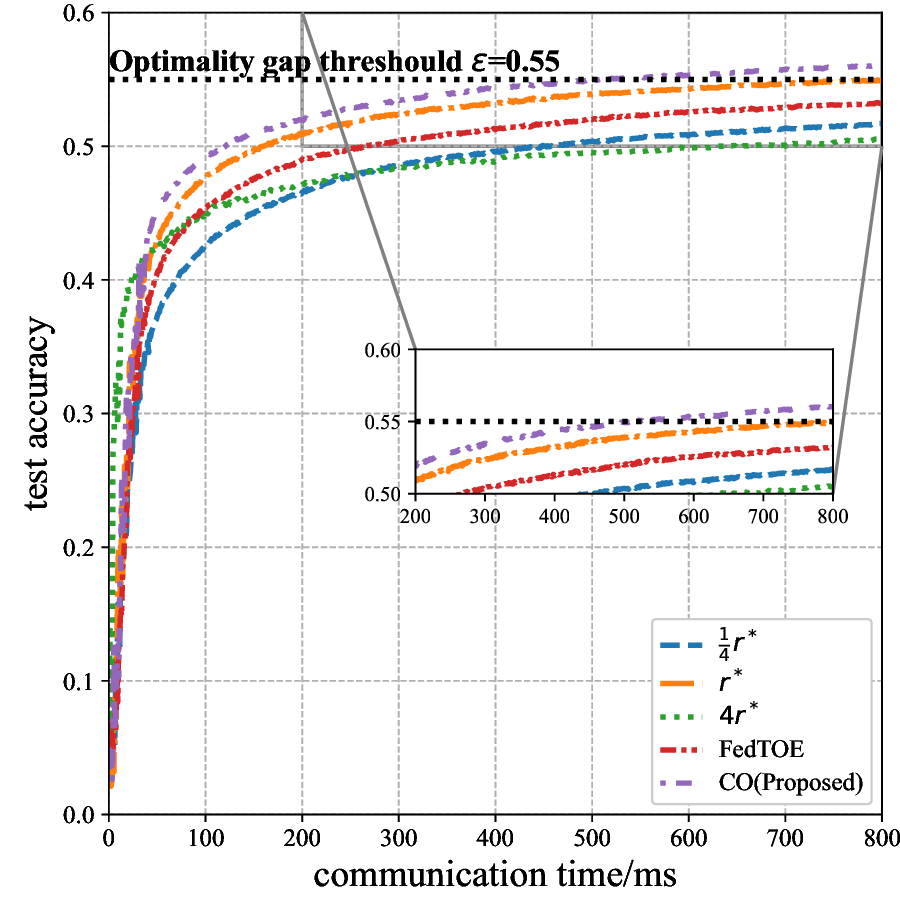}}
        \subfigure[VGG on CIFAR-10]{
            \label{Fig:compression convergence curve:c} 
            \includegraphics[width=0.235\linewidth]{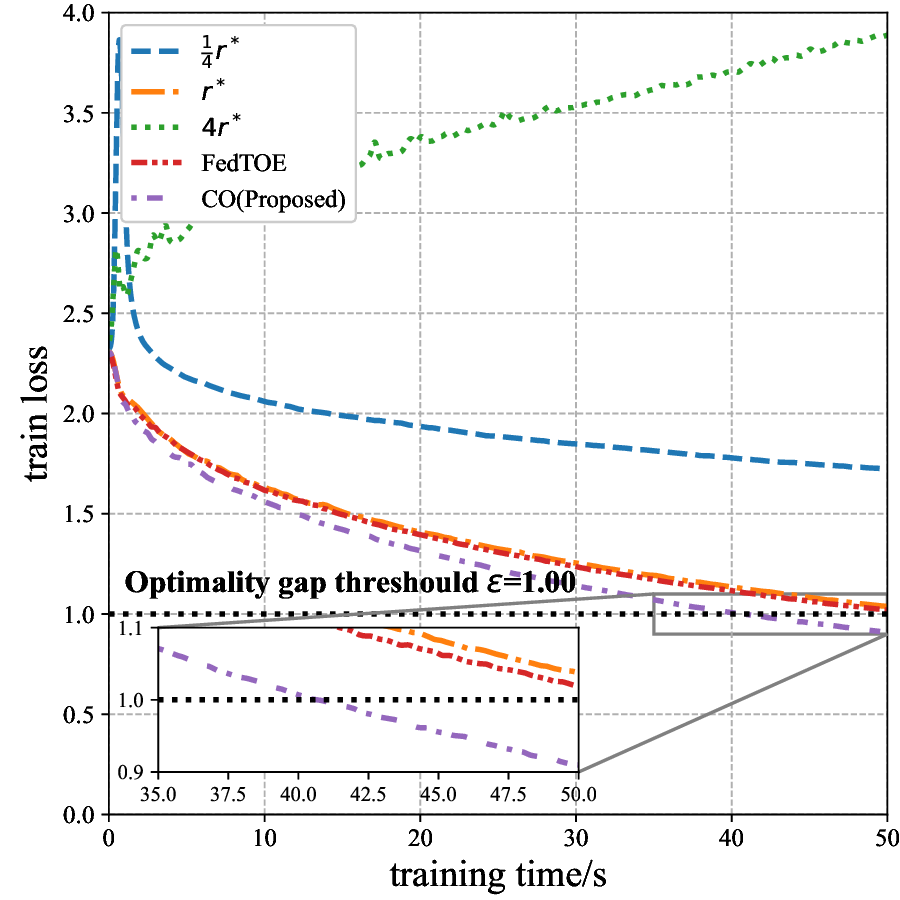}}						    
        \subfigure[VGG on CIFAR-10]{
        \label{Fig:compression convergence curve:d} 
        \includegraphics[width=0.235\linewidth]{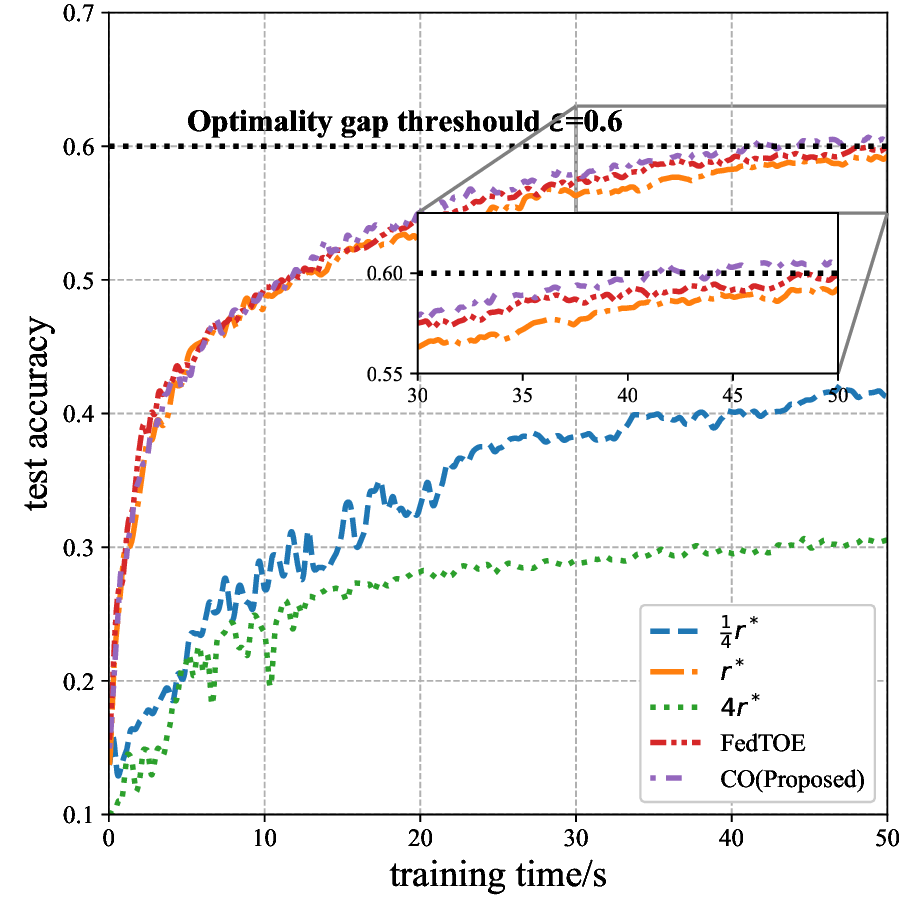}}						
    \caption{Performance comparison of different compression schemes}\label{fig: sparse evalation}
\end{figure*}
\begin{figure*}[h]
	\centering
	\subfigure[\rv{Logistic on FEMNIST}]{
		\label{Fig:convergence curve:a} 
		\includegraphics[width=0.235\linewidth]{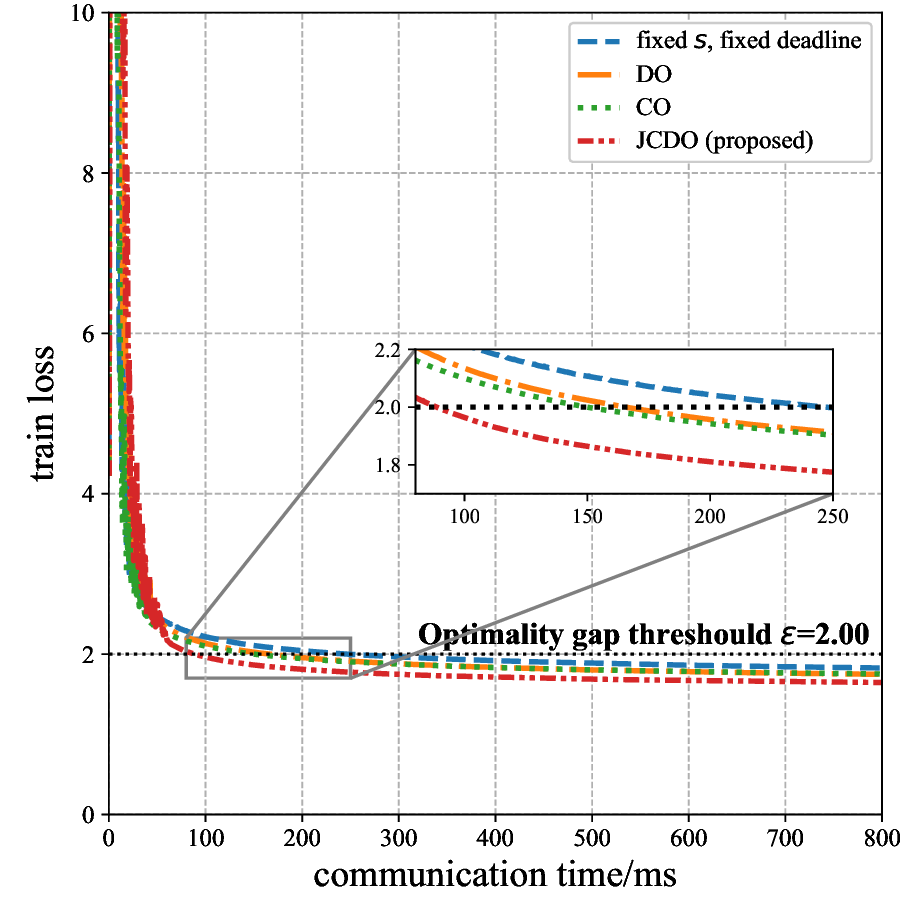}}
    \subfigure[\rv{Logistic on FEMNIST}]{
		\label{Fig:convergence curve:b} 
		\includegraphics[width=0.235\linewidth]{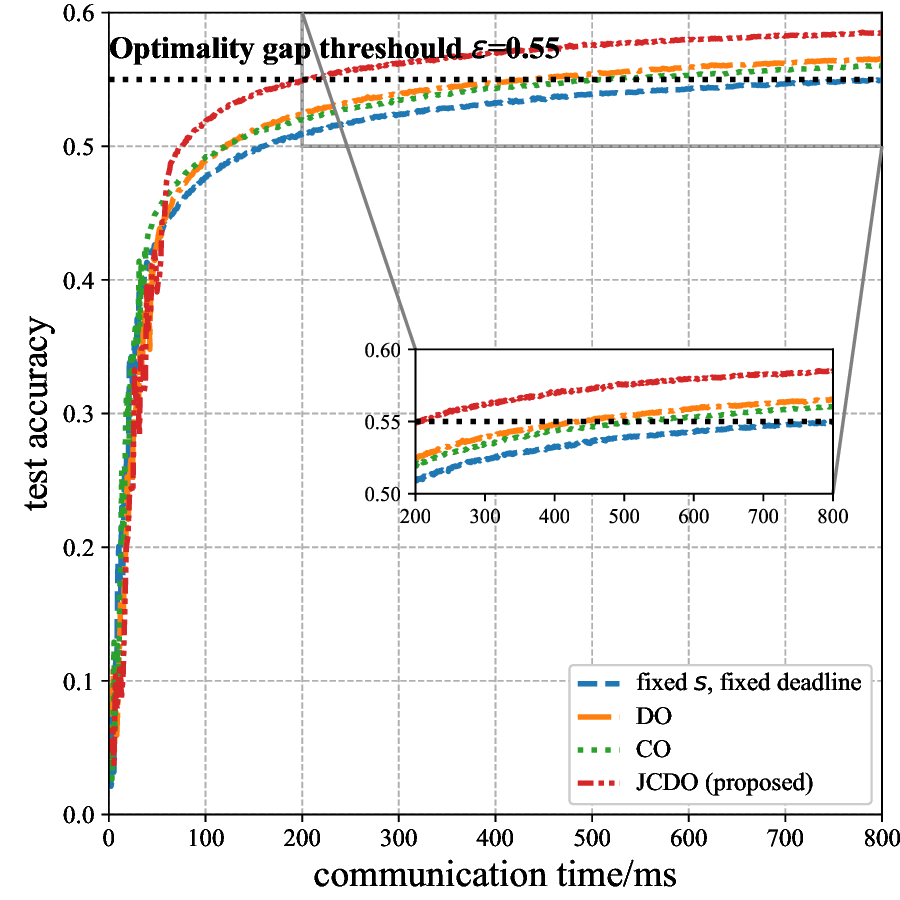}}
        \subfigure[\rv{VGG on CIFAR-10}]{
            \label{Fig:convergence curve:c} 
            \includegraphics[width=0.235\linewidth]{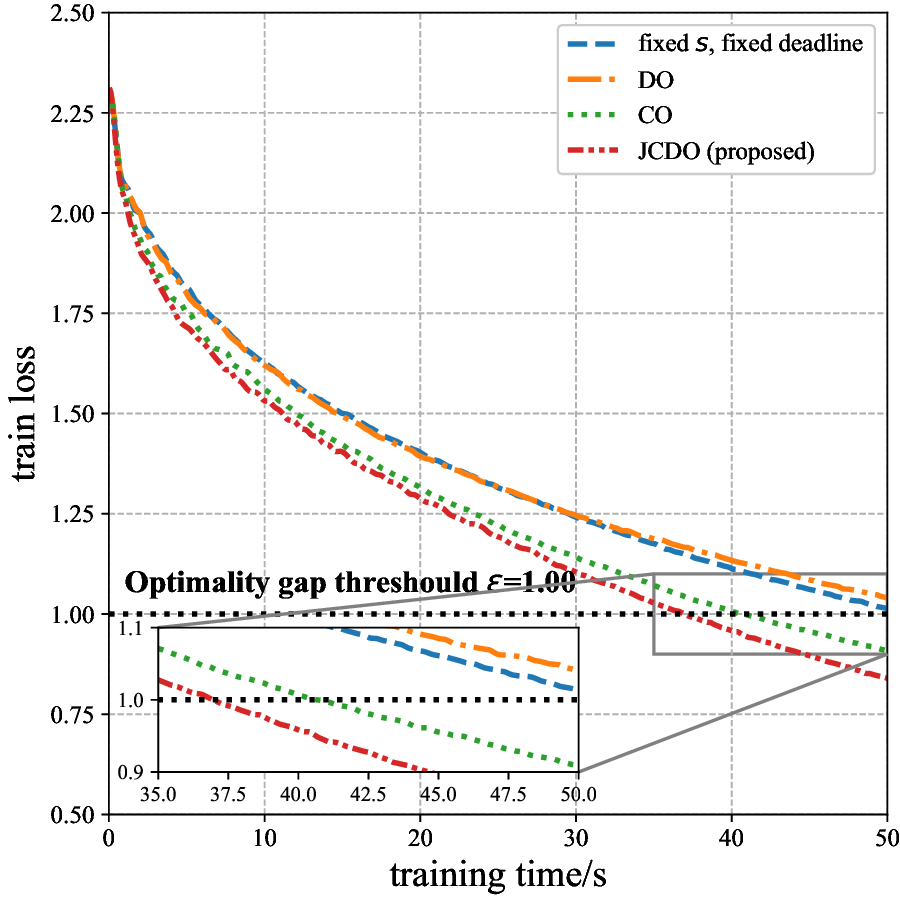}}						    
        \subfigure[\rv{VGG on CIFAR-10}]{
		\label{Fig:convergence curve:d} 
		\includegraphics[width=0.235\linewidth]{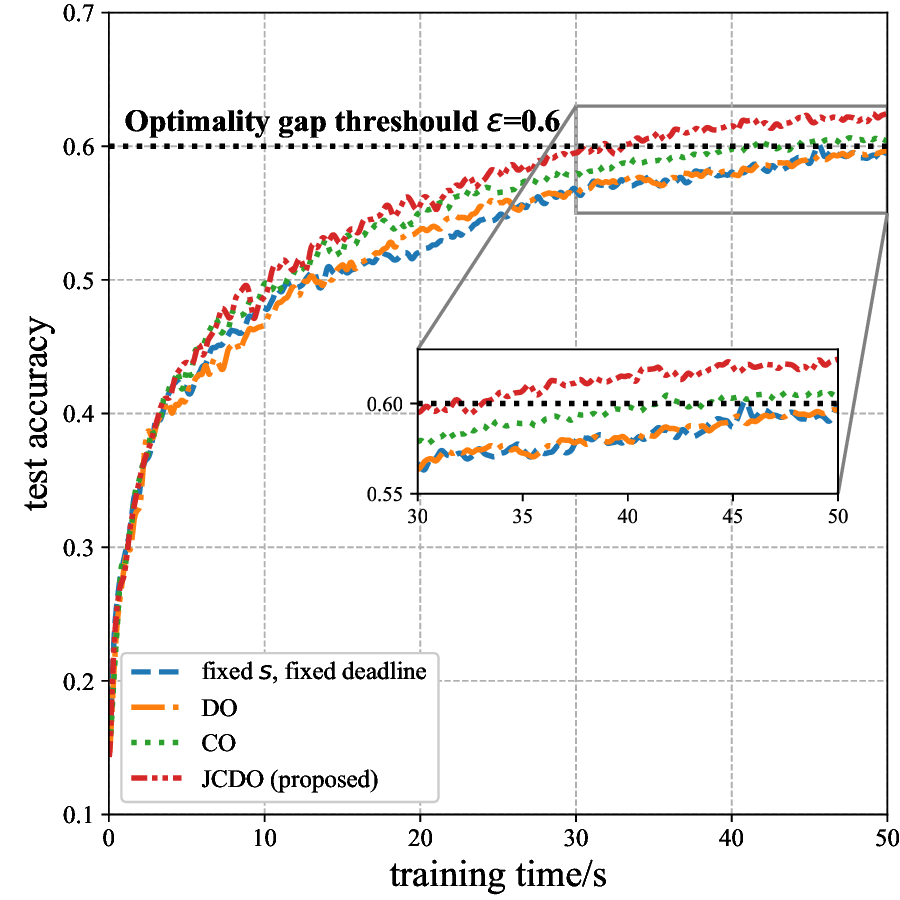}}						
	\caption{\rv{Ablation experiments for evaluating joint optimization algorithms}}\label{Fig:3}
\end{figure*}
\subsubsection{Evaluation of Compression Schemes}\label{subsec: evaluation of compression}
\rv{In this subsection, we explore the impact of various compression ratio $r$ programming schemes, as illustrated in Fig. \ref{fig: sparse evalation}. The deadline here is set to 0.2ms/25ms in the experiment of FEMNIST/CIFAR-10. We evaluate the proposed CO scheme by comparing it with FedTOE\footnote{The parameter $q_m^{(t)}$ in FedTOE is preset to $0.9$.}, and three heuristic schemes that set the same $r$ among devices, namely $r=r^*$, $r=\frac{1}{4}r^*$, $r=4r^*$, where $r^*$ denotes the optimal compression ratio that is adopted across devices and is obtained though numerical experiments.\footnote{$r^*$ foe FEMNIST with Logistic model and CIFAR-10 with VGG model is given by $0.0004$ and $0.0002$.}
  Fig. \ref{Fig:compression convergence curve:a} and Fig. \ref{Fig:compression convergence curve:b} depict the performance metrics of training loss and test accuracy on FEMNIST dataset, respectively. 
  Similarly the results on CIFAR-10 dataset are shown in Fig. \ref{Fig:compression convergence curve:c} and Fig. \ref{Fig:compression convergence curve:d}. 
  Our findings reveal that both higher (i.e., $4r^*$) and lower (i.e., $\frac{1}{4}r^*$) compression ratios impact convergence speed. Specifically, using $4r^*$ maintains high local update precision but makes timely uploading to the BS challenging, thereby increasing the likelihood of transmission outage. On the other hand, employing $\frac{1}{4}r^*$ mitigates the outage issue by necessitating only lightweight updates, but at the cost of introducing substantial compression errors that impair the performance of global model. 
  Consequently, it is crucial to minimize a combined metric that considers both compression error and transmission outage when determining the optimal compression ratio. Moreover, the heuristic $ r $ settings do not adequately address channel heterogeneity across devices. Although FedTOE achieves a relatively good trade-off by adaptively adjusting $ r $ values, it requires pre-setting a common target for transmission outage probability, which may be impractical in real-world FEEL systems.
  In contrast, our proposed CO scheme doesn't require any hyperparameters and optimizes a mixed term to outperform the benchmark schemes in both learning tasks. Specifically, CO accelerates the training process of reaching the goal train loss $\epsilon$ on the FEMNIST dataset by $3.9\times$, $1.6\times$, compared to FedTOE and the best-fixed setting, and similarly outpaces the competition on the CIFAR-10 dataset by $1.2\times$, $1.3\times$. This attests to its superior performance. 
}
\subsubsection{\rv{Evaluation of Joint Optimization Scheme}}
\rv{In this subsection, we evaluate the effectiveness of our proposed joint optimization framework, denoted as JCDO, by contrasting its performance with three benchmark approaches: 1) a fixed compression ratio $ r^* $ combined with an optimal fixed deadline $ T_D^* $, 2) compression ratio optimization with $ T_D^* $ (CO), and 3) deadline optimization with $ r^* $ (DO). Here, $ T_D^* $ is determined through a bisection search method. The comparative results are summarized in Figure \ref{Fig:3}. Specifically, Fig. \ref{Fig:convergence curve:a} and Fig. \ref{Fig:convergence curve:b} showcase the training loss and test accuracy on the FEMNIST dataset using a Logistic Network, respectively. 
Compared to the conventional fixed scheme, JCDO, CO, and DO enhance the training speed to reach a classification accuracy of $ \epsilon=0.55 $ by factors of $ 4.0\times $, $ 1.6\times $, and $ 1.9\times $ respectively. Likewise, our experiments on the CIFAR-10 dataset, as depicted in Fig. \ref{Fig:convergence curve:c} and Fig. \ref{Fig:convergence curve:d}, show that JCDO and CO expedite the training to achieve a classification accuracy of $ \epsilon=0.55 $ by $ 1.4\times $ and $ 1.1\times $, when compared to the fixed scheme. 
Through extensive ablation studies, we validate the effectiveness of both the compression and deadline optimization strategies, as well as their synergistic effects when jointly optimized. Thus, our joint optimization framework, JCDO, provides a compelling solution for improving the training efficiency in FEEL. 
}
					
\subsubsection{{Evaluation on Resources}}
{
    In this subsection, we compare the proposed JCDO algorithm with the benchmark schemes 
    given different bandwidth and number of devices.  
    Under a given training objective (i.e., the test accuracy achieves 0.5), the required communication time is depicted in Fig. \ref{Fig:bandwidth comparison}.  
    It can be found that system bandwidth has a huge impact on the required training time. 
    More specifically, the proposed CO method outperforms FedTOE since CO minimizes the aggregation error introduced by data compression and transmission outage. 
    Moreover,  the DO method that optimizes the per-round deadline outperforms FedTOE and has a comparable performance with CO, which illustrates the importance and necessity of deadline setting. 
    Finally, the proposed JCDO scheme attains the fastest convergence in all the examined bandwidth settings, which validates the effectiveness of the proposed scheme. 
    }

{
    The performance comparison given various number of devices is illustrated in Fig. \ref{fig: number of devices}, where all the schemes have the same training time budget (i.e., 100ms). 
    After training, we evaluate the global model obtained through different schemes. When number of devices is small (i.e., 10), the randomness caused by transmission outage is large and more likely to result in biased aggregation,  which eventually renders performance degradation. 
    As the number of devices increases, the performance of the scheme based on data compression and deadline setting gradually improves.
     Compared to FedTOE, the two baseline algorithms, i.e., CO and DO,  demonstrate some performance gain, and the proposed JCDO algorithm outperforms CO and DO, which further validates the effectiveness and robustness of the proposed algorithm. 
}
\begin{figure}[h]
    \centering
    \includegraphics[width=0.8\linewidth]{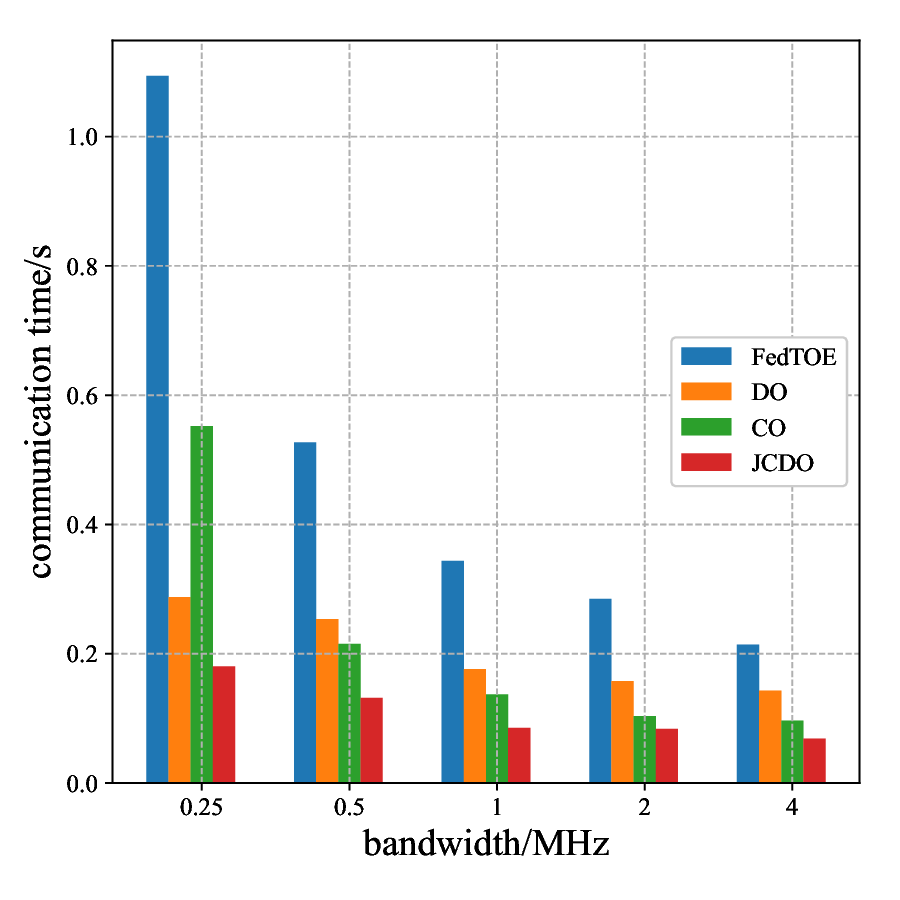}
    \caption{{Performance comparison of different bandwidth settings under a given training time (Logistic on FEMNIST)}}
    \label{Fig:bandwidth comparison}
\end{figure}
\begin{figure}[h]
    \centering
    \includegraphics[width=0.8\linewidth]{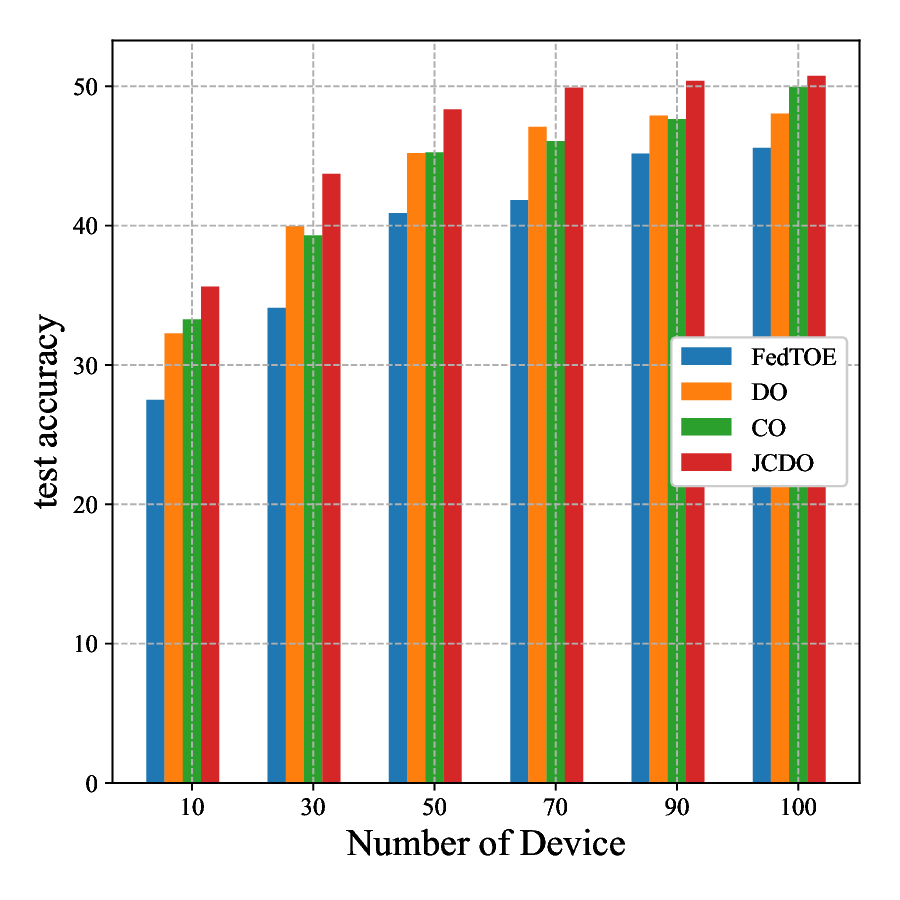}
    \caption{{Performance comparison of different device number under a given training time (Logistic on FEMNIST)}}\label{fig: number of devices}
    \label{Fig:device number comparison}
\end{figure}
\subsection{\rv{FEEL for Advanced Task: Object Detection in Autonomous Driving}}
\begin{figure}[h!]
    \centering
    \includegraphics[width=0.8\linewidth]{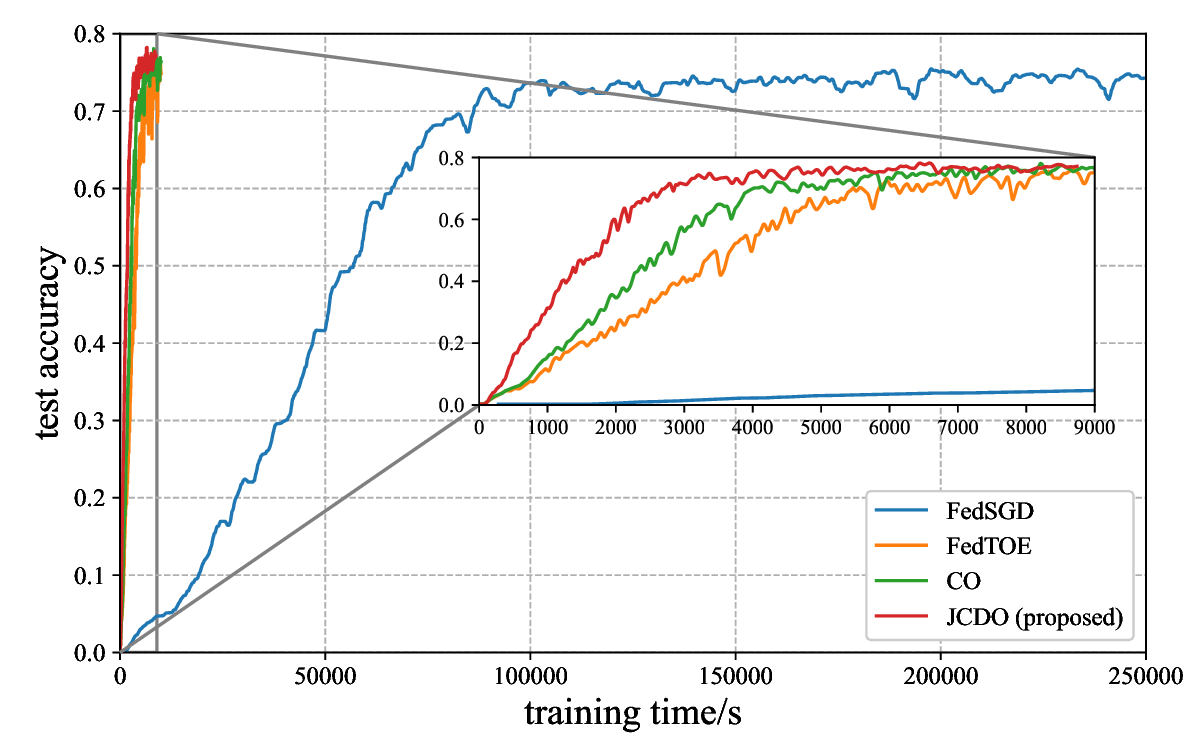}
    \caption{Performance comparison of different compression schemes (Object detection on Kitti dataset)}\label{fig: yolo sparse evalation}
    \label{Fig:AD test acc}
\end{figure}
\begin{figure*}[h]
    \centering
    \includegraphics[width=1\linewidth]{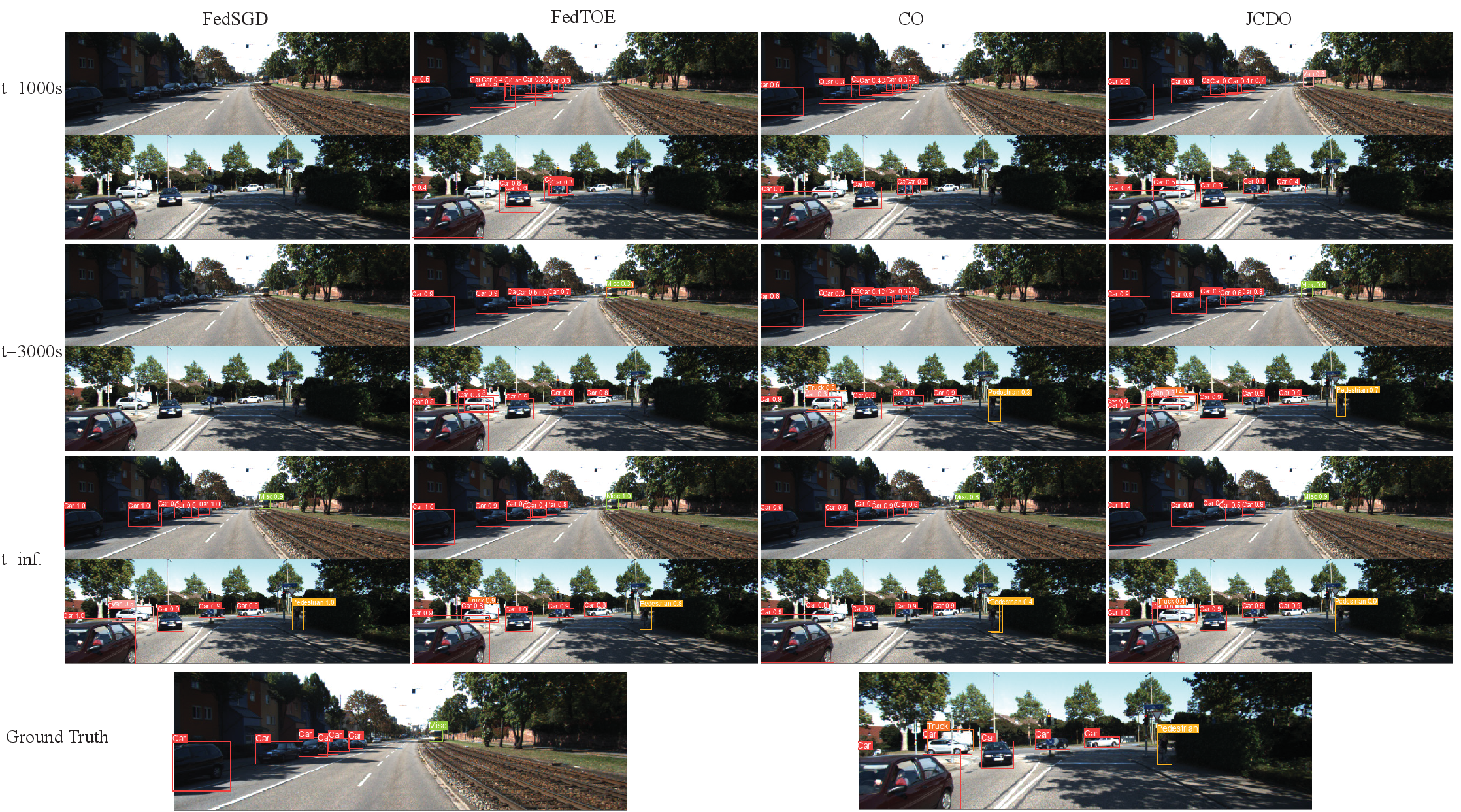}
    \caption{Detection results of different schemes under three training stages}
    \label{Fig:AD detection results}
\end{figure*}
In this subsection, we compare the proposed JCDO scheme with three benchmark schemes in the task of object detection, where the $q_m^{(t)}$ in FedTOE is preset to $0.5$. 
The test performance versus training time is shown in Fig. \ref{Fig:AD test acc}. 
Firstly, it can be seen that the original FedSGD algorithm requires an extremely long time to complete the learning task. 
This is because the FEEL requires a time-consuming data exchange during each update round. 
Devices need to upload $7,041,025$ parameters in YOLOv5 network to the server simultaneously, 
which results in prohibitive communication overhead for bandwidth limited communication system. 
In addition, due to the instability of the wireless environment, some straggler devices further prolong the per-round latency. 
It reminds us that data compression and deadline design are necessary for the practical deployment of FEEL. 
Secondly, the total training time of the other three schemes based on data compression is significantly shortened, all of them can complete the task in $9000$s. 
Finally, the proposed JCDO scheme still outperforms CO and FedTOE, as it gives better guidance for the setting of sparse ratio and deadline based on theoretical analysis.  
It can be found that the proposed JCDO scheme achieves a speed-up of 30 times compared to the vanilla FedSGD algorithm. 
The robustness of the proposed scheme for tackling a variety of learning tasks is verified. 

As shown in Fig. \ref{Fig:AD detection results}, we directly evaluate the schemes from the detection performance of the model obtained at different training stages. 
In the early training stage ($t=1000$s), FedSGD can not detect any object, and CO and FedTOE scheme detect some cars, while the proposed JCDO scheme detects all the cars. 
In the middle training stage ($t=3000$s), FedSGD still cannot detect the object since it suffers from too few update rounds. 
FedTOE and CO can detect most of the cars at this time, while some small object classes, such as pedestrian and misc,  are missed. 
Fortunately, the proposed JCDO scheme can detect them successfully. 
At the later training stage (i.e., with a sufficient long training time), all the schemes have good detection performance. 
We can conclude now that the proposed JCDO scheme is most communication-efficient among these schemes. 
\section{Conclusion}
In this paper, we proposed a communication-efficient FEEL algorithm called JCDO, which alleviates the communication burden through local update compression  and deadline for the uploading process. 
The compression ratio and the deadline are designed for remaining training time minimization, and obtained through alternating iterative optimization. 
Experiments show that the proposed scheme can greatly accelerate FEEL, 
and outperform the existing schemes. 
Future works will focus on the extensions to the scenarios of  model averaging and the aggregation weights optimization. 
\appendices
\section{Proof of Lemma 1} \label{subsec: proof of lemma 2}
With the random sparse operator in (3), for each element in $\mathbf{g}$, we have
\begin{align}
	\mathbb{E}\left[\mathcal{S}\left(\mathbf{g}\right)\right] &= \left[\mathbb{E}\left[Z_1\frac{g_1}{p_1}\right],\mathbb{E}\left[Z_2\frac{g_2}{p_2}\right],...,\mathbb{E}\left[Z_S\frac{g_S}{p_S}\right]\right]\notag\\&=[g_1,g_2,...,g_S]=\mathbf{g}, 
\end{align}
Therefore, the unbiasedness has been proved. 

According to the optimal solution of $\mathscr{P}_1$, $p_i^*=\min\left\{\frac{\left|g_i\right|}{\lambda},1\right\}$, $p_i^*$ is proportional to $\left|g_i\right|$ when $p_i^*<1$. 
Without loss of generality, we sort $\left\{g_i\right\}$ by an decreasing order of $\left|g_i\right|$, i.e., $\left|g_{(1)}\right|\geq\left|g_{(2)}\right|\geq...\geq\left|g_{(S)}\right|$. 
It can be verified that there exists an index $j$ such that $p_{(j)}=1$ if $i\leq j$, and $p_{(i)}=\frac{|g_i|}{\lambda}$ if $j<i\leq D$. 
With the primal feasibility condition $\frac{\sum_{i=1}^S p_{(i)}}{S}=r$, we have 
\begin{align}\label{eq:lambda close form solution}
	\lambda = \frac{\sum_{i=j+1}^S \left|g_{(i)}\right|}{Sr-j}. 
\end{align}
Then we can obtain the variance of the compressor as follows. 
\begin{align}
	\mathbb{E}&\left\|{\rm Comp}\left(\mathbf{g},r\right)-\mathbf{g}\right\|_2^2\notag\\&=\sum_{n=1}^j \mathbb{E}\left|Z_{(n)} \frac{g_{(n)}}{p_{(n)}}-g_{(n)}\right|^2+\sum_{i=j+1}^S \mathbb{E}\left|Z_{(i)}\frac{g_{(i)}}{p_{(i)}}-g_{(i)}\right|^2\notag\\
&=\sum_{i=j+1}^S\left|g_{(i)}\right|^2\left(\frac{1}{p_{(i)}}-1\right)\label{eq: variance middle results}\\
&=\frac{\left(\sum_{i=j+1}^S\left|g_{(i)}\right|\right)^2}{Sr-j}-\sum_{i=j+1}^S\left|g_i\right|^2\label{eq:middle equation of comp}\\
&=\left(\frac{\left\|\bar{\mathbf{g}}_j\right\|_1^2}{\left(Sr-j\right)\left\|\bar{\mathbf{g}}_j\right\|_2^2}-1\right)\frac{\left\|\bar{\mathbf{g}}_j\right\|_2^2}{\left\|\mathbf{g}\right\|_2^2}\left\|\mathbf{g}\right\|_2^2\label{eq: variance final expression}\\
&=F\left(r,\left\{\left|g_i\right|\right\}\right)\left\|\mathbf{g}\right\|_2^2,
\end{align}
where (\ref{eq: variance middle results}) comes from the fact that $p_{(n)}=1$ for $n=1,\cdots, j$. 
(\ref{eq:middle equation of comp}) is due to $p_{(i)}=\frac{|g_i|}{\lambda}$, and $\lambda$ is given in (\ref{eq:lambda close form solution}), 
and $\bar{\mathbf{g}}_j$ is defined as  $\bar{\mathbf{g}}_j=\left[g_{(j+1)},...,g_{(S)}\right]$. 

It is found in (\ref{eq: variance final expression}) that the variance result is a complex function of $r$ and $\left\{{g}_i\right\}$. 
To obtain an estimated expression of $F\left(r,\left\{\left|g_i\right|\right\}\right)$, 
we first relax $\mathscr{P}_1$, that is, discarding $p_i \leq 1, \forall i$. 
Then the corresponding $p_i^{**}$ is given by
\begin{align}
	p_i^{**} = \frac{Sr\left|g_i\right|}{\sum_{i=1}^S\left|g_i\right|}.
\end{align} 
We define the compressor with suboptimal solution as ${\rm Comp}'\left(\mathbf{g},s\right)$, 
the variance is given by
\begin{align}
	\mathbb{E}\left\|{\rm Comp'}\left(\mathbf{g},r\right)-\mathbf{g}\right\|&=\sum_{i=1}^S\left|g_{i}\right|^2\left(\frac{1}{p_i}-1\right)\notag\\
&=\left(\frac{\left(\sum_{i=1}^S\left|g_i\right|\right)^2}{S\sum_{i=1}^S\left|g_i\right|^2}\frac{1}{r}-1\right)\left\|\mathbf{g}\right\|_2^2 \notag\\ 
&=\left(\frac{\left\|\mathbf{g}\right\|_1^2}{S\left\|\mathbf{g}\right\|_2^2}\frac{1}{r}-1\right)\left\|\mathbf{g}\right\|_2^2 \notag \\
&=F'\left(r,\left\{g_i\right\}\right)\left\|\mathbf{g}\right\|_2^2.
\end{align}
\begin{figure}[!]
	\centering
	\includegraphics[width=0.8\linewidth]{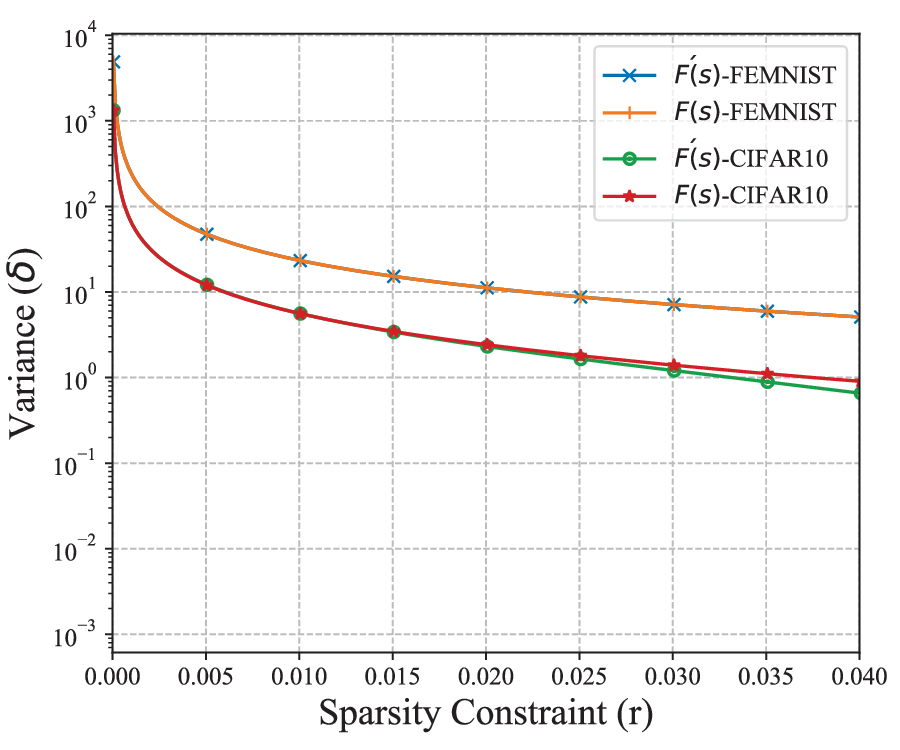}					
	\caption{Variance Coefficient Estimation}
	\label{Fig:sparse evaluation}
\end{figure}
We compare $F'\left(r,\left\{g_i\right\}\right)$ with $F\left(r,\left\{g_i\right\}\right)$ from the experiments, e.g., the classification on FEMNIST and CIFAR10 dataset. As shown in Fig. \ref{Fig:sparse evaluation},  
$F\left(r,\left\{g_i\right\}\right)$ and $F'\left(r,\left\{g_i\right\}\right)$ are almost equivalent in both experiments.\footnote{Note that, we performed comparative experiments under $r\in(0,0.04]$. This is because the gradient itself usually has a large degree of sparsity, according to our experiments, when $r$ is greater than $0.04$, sparsity operation basically has little impact on the training process.}
 Thus we use $F'\left(r,\left\{g_i\right\}\right)$ to estimate $F\left(r,\left\{g_i\right\}\right)$, which ends the proof. 
\section{Proof of Proposition 1} \label{subsec: proof of proposition 1}
The proof of Proposition 1 follows the widely-adopted analysis of convergence rate. 
Firstly, the following lemma is needed. 
\begin{lemma} \label{eq: per round improvement}
    \emph{    With the stepsize $\eta^{(n)}=\frac{\chi}{n+\nu}$, data compression error $\left\{\delta_m^{(n)},\forall m,n\geq t\right\}$, successful transmission probability $\left\{q_m^{(n)},n\geq t\right\}$, the improvement of $n$-round $(n\geq t)$ is bounded by}
    \begin{align}\label{eq: one round improvement}
        \mathbb{E}\left\|\mathbf{w}^{(n+1)}-\mathbf{w}^*\right\|^2\leq \left(1-\frac{3\mu}{2}\eta^{(n)}\right)\left\|\mathbf{w}^{(n)}-\mathbf{w}^*\right\|^2 \notag\\+\left(\eta^{(n)}\right)^2D_n.
    \end{align}
    \emph{where} $D_n=G\sum_{m=1}^M\frac{d_m^2}{d^2} \left(\frac{1+\delta_m^{(n)}}{q_m^{(n)}}-1\right) +  \sum_{m=1}^M\frac{d_m^2}{d^2} \sigma^2$, $\delta_m^{(n)}=\frac{a_m^{(n)}}{r_m^{(n)}}-1$. 
\end{lemma}
\begin{proof}
    See Appendix \ref{proof of lemma 5}. 
\end{proof}
Based on Lemma \ref{eq: per round improvement}, let $\Delta_n = \mathbb{E} \left\|\mathbf{w}^{(n)}-\mathbf{w}^*\right\|^2$, we have
\begin{align}
    \Delta_{n+1}\leq\left(1-\frac{3\mu}{2}\frac{\chi}{n+\nu}\right)\Delta_n+\frac{\chi^2}{(n+\nu)^2}D_n, n\geq t.  
\end{align}
Here we first proof the following inequality
\begin{align} \label{eq: recursive inequality}
\Delta_{t+l} \leq \frac{\zeta^{(t)}}{t+\nu+l}, 
\end{align}
where $\zeta^{(t)}=\max\left\{\frac{2\chi^2\Xi_t}{{3\mu\chi}-2},\left(t+\nu\right)\left\|\mathbf{w}^{(t)}-\mathbf{w}^*\right\|^2\right\}$, $\Xi_t=G\sum_{m=1}^M\frac{d_m^2}{d^2} \left(\max_{n=t,...,l}\frac{1+\delta_m^{(n)}}{q_m^{(n)}}-1\right) +  \sum_{m=1}^M\frac{d_m^2}{d^2} \sigma^2$.

For $l=1$, we have $\triangle_{t}=\frac{(t+\nu)\triangle_{t}}{t+\nu}\leq\frac{\zeta^{(t)}}{t+\nu}$. Therefore, (\ref{eq: recursive inequality}) holds for $l=0$. 

Then we assume (\ref{eq: recursive inequality}) holds for some $l=r (r\geq0)$, i.e., $\triangle_{t+r}\leq\frac{\zeta^{(t)}}{t+\nu+r}$. 

For $l=r+1$, it follows that 
\begin{align}
    \triangle_{t+r+1}&\leq\left(1-\frac{3\mu\chi}{2(t+r+\nu)}\right)\triangle_{t+r}+\frac{D_{t+r}\chi^2}{(t+r+\nu)^2} \notag\\
    &\leq\left(1-\frac{3\mu\chi}{2(t+r+\nu)}\right)\frac{\zeta^{(t)}}{t+r+\nu}+\frac{D_{t+r}\chi^2}{(t+r+\nu)^2} \notag\\
    &\leq\frac{2(t+r+\nu-1)\zeta^{(t)}-(3\mu\chi-2)\zeta^{(t)}+2D_{t+r}\chi^2}{2(t+r+\nu)^2} \notag\\
    &\leq\frac{(t+r+\nu-1)\zeta^{(t)}}{(t+r+\nu)^2} \label{eq:the comparison between Xi and D}\\
    &\leq\frac{\zeta^{(t)}}{t+r+\nu+1},  \notag
    \end{align}
where (\ref{eq:the comparison between Xi and D}) is because $\Xi_t\geq D_n, \forall n\geq t$.     
It is found that (\ref{eq: recursive inequality}) holds for $l=r+1$ when holds for $l=r$. By now we can conclude that (\ref{eq: recursive inequality}) holds for $\forall l=0,1,2..., N_{\epsilon}^{(t)}$. 

We can further have 
\begin{align}
    \mathbb{E}\left\|\mathbf{w}^{(t+l)}-\mathbf{w}^*\right\|^2&\leq \frac{\zeta^{(t)}}{t+\nu+l}\notag\\&\leq\frac{(t+\nu)\left\|\mathbf{w}^{(t)}-\mathbf{w}^*\right\|^2+ \frac{2\chi^2\Xi_t}{{3\mu\chi}-2}}{t+\nu+l}.
\end{align}
Then with the Assumption 1 and 2, we can conclude that 
\begin{align}
    \mathbb{E}&\left(L\left(\mathbf{w}^{(t+l)}\right)-L\left(\mathbf{w}^*\right)\right)\notag\\&\leq \frac{\ell}{2}\mathbb{E}\left\|\mathbf{w}^{(t+l)}-\mathbf{w}^*\right\|^2 \notag\\
    &\leq\frac{\frac{\ell(t+\nu)}{2}\left\|\mathbf{w}^{(t)}-\mathbf{w}^*\right\|^2+\frac{\ell\chi^2\Xi_t}{{3\mu\chi}-2}}{t+\nu+l} \notag \\
    &\leq \frac{\frac{\ell(t+\nu)}{\mu}\left(L\left(\mathbf{w}^{(t)}\right)-L\left(\mathbf{w}^*\right)\right)+\frac{\ell\chi^2\Xi_t}{{3\mu\chi}-2}}{t+\nu+l}, 
\end{align}
which ends the proof. 
\section{Proof of Lemma 3}\label{appendix: proof of remaining communicaiton rounds}
Based on Proposition 1,  $\{r_m^{n}=r_m^{(t)},n\geq t\}$,  the deadline $\{T_D^{(n)}=T_D^{(t)}, n\geq t\}$ and Assumption 5, we have 
\begin{align}\label{eq: intermideate result of convergence analysis}
    \mathbb{E}&\left(L\left(\mathbf{w}^{(t+l)}\right)-L\left(\mathbf{w}^*\right)\right)\notag\\
	&\leq \frac{1}{t+\nu+l}\bigg({{\frac{\ell(t+\nu)}{\mu}\left(L\left(\mathbf{w}^{(t)}\right)-L\left(\mathbf{w}^*\right)\right)}}\notag\\&~~~~~~~~~~~~~~+{A\sum_{m=1}^M\frac{d_m^2}{d^2}\sigma^2}+{AG\sum_{m=1}^M\big(\frac{\alpha_m^{(t)}}{r_m^{(t)}q_m^{(t)}}-1\big)}\bigg), 
\end{align}
Let the right hand sight of (\ref{eq: intermideate result of convergence analysis}) equal to $\epsilon$, then we have  $\mathbb{E}\left(L\left(\mathbf{w}^{(t+l)}\right)-L\left(\mathbf{w}^*\right)\right) \leq \epsilon$. 
Therefore, the expectation of $N_{\epsilon}^{(t)}$ should meet 
\begin{align}
    N_{\epsilon}^{(t)} &\leq l = {\frac{\ell(t+\nu)}{\mu\epsilon}\left(L\left(\mathbf{w}^{(t)}\right)-L\left(\mathbf{w}^*\right)\right)}-t-\nu \notag\\~~~~&+\frac{A}{\epsilon}\sum_{m=1}^M\frac{d_m^2}{d^2}\sigma^2 + \frac{AG}{\epsilon}\sum_{m=1}^M\left(\frac{\alpha_m^{(t)}}{r_m^{(t)}q_m^{(t)}}-1\right). 
\end{align}
It is the same as Lemma 3, which ends the proof. 
\section{Proof of Proposition 2}\label{subsecL proof of proposition 2}
Substitute (11) into (21), $\mathscr{P}_4$ can be rewritten as follow. 
\begin{align}
	\mathscr{P}_2: \min_{s_m^{(t)}}~&
	\frac{1}{r_m^{(t)}}e^{\frac{BN_0\left(2^{\frac{bSr_m^{(t)}}{BT_{D,m}^{(t)}}}-1\right)}{P_m\sigma_m^2}},\\
	{\rm s.t.}~& 0< r_m^{(t)}\leq 1. \notag
\end{align}
Take logarithm of objective function and remove irrelevant term, the equivalent objective function is given by
\begin{align}
	g\left(r_m^{(t)}\right)=\frac{BN_0}{P_m\sigma_m^2}\left(2^{\frac{bSr_m^{(t)}}{BT_{D,m}^{(t)}}}-1\right) - \ln\left(r_m^{(t)}\right). 
\end{align}
Obviously, $g\left(r_m^{(t)}\right)$ is convex, let $g'\left(r_m^{(t)}\right)=0$, the desired results can be obtained,which ends the proof.  

\section{Proof of Lemma 4}\label{subsec:proof of convex lemma}
To proof Lemma 4, we only needs to proof that the following function is convex. 
\begin{align}
	h\left(x\right)=xe^{\frac{2^{\frac{c}{x-b}}-1}{y}},
\end{align}
where $x\geq b$ and $b,c,y \geq 0$.

The second derivative of $h(x)$ is given by
\begin{align}
	h''\left(x\right) &= \ln\left(2\right)\,c2^\frac{c}{x-b}\mathrm{e}^\frac{2^\frac{c}{x-b}-1}{y}\cdot\notag\\ &\dfrac{\left(\ln\left(2\right)\,cx2^\frac{c}{x-b}+\ln\left(2\right)\,cyx+2by(x-b)\right)}{y^2\cdot\left(x-b\right)^4}
\end{align}
It is found that $h''\left(x\right)\geq0$ holds for $\forall x \geq b$. Therefore, $h(x)$ ix convex, which ends the proof. 

\section{Proof of Lemma \ref{eq: per round improvement}}\label{proof of lemma 5}
We first introduce some additional notations as follows. 
\begin{align}
    \mathbf{w}^{(n+1)}&=\mathbf{w}^{(n)}-\eta^{(n)}\sum_{m=1}^M\frac{d_m}{d}\frac{Y_m^{(n)}}{q_m^{(n)}} {\rm Comp}\left({\mathbf{g}}_m^{(n)},r_m^{(n)}\right),\\
    \overline{\mathbf{w}}^{(n+1)}&=\mathbf{w}^{(n)}-\eta^{(n)}\sum_{m=1}^M\frac{d_m}{d}{\rm Comp}\left({\mathbf{g}}_m^{(n)},r_m^{(n)}\right),\\
    \widetilde{\mathbf{w}}^{(n)}&=\mathbf{w}^{(n)}-\eta^{(n)}\sum_{m=1}^M\frac{d_m}{d}{\mathbf{g}}_m^{(n)}.
\end{align}
Then based on the property of the compressor and unbiased aggregation, we have the following Lemmas. 
\begin{lemma}
\emph{(Unbiasedness) Consider a specific round $n=0,1,...,N$, the expectation gap between $\mathbf{w}^{(t+1)}$ and $\mathbf{w}^*$ can be decomposed as follows. }
\begin{align} \label{eq: decompose of gap}
    \mathbb{E}&\left\|\mathbf{w}^{(n+1)}-\mathbf{w}^*\right\|^2=\mathbb{E}\left\|{\mathbf{w}}^{(n+1)}-\overline{\mathbf{w}}^{(n+1)}\right\|^2\notag\\&~~~~~~+\mathbb{E}\left\|\overline{\mathbf{w}}^{(n+1)}-\widetilde{\mathbf{w}}^{(n+1)}\right\|^2+\mathbb{E}\left\|\widetilde{\mathbf{w}}^{(n+1)}-\mathbf{w}^*\right\|^2. 
\end{align}
\end{lemma}
\begin{proof}
    The proof of (\ref{eq: decompose of gap}) can be obtained through the following unbiased properties. 
    \begin{align}
        \mathbb{E}_{Y_m^{(n)}}\left[\mathbf{w}^{(n+1)}\right]&=\overline{\mathbf{w}}^{(n+1)},\\
        \mathbb{E}_{\rm Comp}\left[\overline{\mathbf{w}}^{(n+1)}\right]&=\widetilde{\mathbf{w}}^{(n+1)}.
    \end{align}
\end{proof}
\begin{lemma}
\emph{(Effect of outage) The expectation gap between $\mathbf{w}^{(n+1)}$ and $\overline{\mathbf{w}}^{(n+1)}$ can be bounded by } 
\begin{align}\label{eq: effect of outage}
    \mathbb{E}&\left\|{\mathbf{w}}^{(n+1)}-\overline{\mathbf{w}}^{(n+1)}\right\|^2\notag\\&\leq\left(\eta^{(n)}\right)^2\sum_{m=1}^M\frac{d_m^2}{d^2} \left(\frac{1}{q_m^{(n)}}-1\right)\left(1+\delta_m^{(n)}\right)\left\|{\mathbf{g}}_m^{(n)}\right\|^2.
\end{align}
\end{lemma}
\begin{proof}
    \begin{align}
        \mathbb{E}&\left\|{\mathbf{w}}^{(n+1)}-\overline{\mathbf{w}}^{(n+1)}\right\|^2\\&=\left(\eta^{(n)}\right)^2\mathbb{E}\left\|\sum_{m=1}^M\frac{d_m}{d}\left(\frac{Y_m^{(n)}}{q_m^{(n)}}-{1}\right) {\rm Comp}\left({\mathbf{g}}_m^{(n)},r_m^{(n)}\right)\right\|^2 \notag\\
        &=\left(\eta^{(n)}\right)^2\sum_{m=1}^M\bigg[\frac{d_m^2}{d^2}\mathbb{E}_{Y_m^{(n)}}\left(\frac{Y_m^{(n)}}{q_m^{(n)}}-1\right)^2 \notag\\&\qquad\qquad\qquad\qquad\qquad\qquad\mathbb{E}_{\rm Comp}\left\|{\rm Comp}\left({\mathbf{g}}_m^{(n)},r_m^{(n)}\right)\right\|^2 \bigg]\notag\\
        &\leq \left(\eta^{(n)}\right)^2\sum_{m=1}^M\frac{d_m^2}{d^2} \left(\frac{1}{q_m^{(n)}}-1\right)\left(1+\delta_m^{(n)}\right)\left\|{\mathbf{g}}_m^{(n)}\right\|^2. \notag
        \end{align}
\end{proof}
\vspace{-4mm}
\begin{lemma}
    \emph{(Effect of compression) The expectation gap between $\overline{\mathbf{w}}^{(n+1)}$ and $\widetilde{\mathbf{w}}^{(n+1)}$ can be bounded by}
    \begin{align}\label{eq:effect of compression}
        \mathbb{E}\left\|\overline{\mathbf{w}}^{(n+1)}-\widetilde{\mathbf{w}}^{(n+1)}\right\|^2\leq\left(\eta^{(n)}\right)^{2}\sum_{m=1}^M\frac{d_m^2}{d^2}\delta_m^{(n)}\left\|{\mathbf{g}}_m^{(n)}\right\|^2.
    \end{align}
\end{lemma}
\vspace{-1mm}
\begin{proof}
    \begin{align}
        \mathbb{E}&\left\|\overline{\mathbf{w}}^{(n+1)}-\widetilde{\mathbf{w}}^{(n+1)}\right\|^2\notag\\&=\left(\eta^{(n)}\right)^2\mathbb{E}\left\|\sum_{m=1}^M\frac{d_m}{d}\left({\rm Comp}\left({\mathbf{g}}_m^{(n)},r_m^{(n)}\right)-{\mathbf{g}}_m^{(n)}\right)\right\|^2 \notag\\
        &=\left(\eta^{(n)}\right)^2\sum_{m=1}^M\frac{d_m^2}{d^2}\mathbb{E}\left\|{\rm Comp}\left({\mathbf{g}}_m^{(n)},r_m^{(n)}\right)-{\mathbf{g}}_m^{(n)}\right\|^2 \notag\\
        &\leq\left(\eta^{(n)}\right)^{2}\sum_{m=1}^M\frac{d_m^2}{d^2}\delta_m^{(n)}\left\|{\mathbf{g}}_m^{(n)}\right\|^2. \notag
        \end{align}
\end{proof}
\begin{lemma}
    \emph{(Effect of SGD) The expectation gap between $\widetilde{\mathbf{w}}^{(n+1)}$ and $\mathbf{w}^*$ can be bounded by}
    \begin{align}\label{eq:effect of SGD}
        \left\|\widetilde{\mathbf{w}}^{(n+1)}-\mathbf{w}^*\right\|^2\leq(1-\frac{3\mu}{2}\eta^{(n)})\left\|\mathbf{w}^{(n)}-\mathbf{w}^*\right\|^2 \notag\\\qquad\qquad\qquad\qquad\qquad+ \left(\eta^{(t)}\right)^2\sum_{m=1}^M\frac{d_m^2}{d^2} \sigma^2.
    \end{align}
\end{lemma}
\vspace{-1mm}
\begin{proof}
    \begin{align}
        \mathbb{E}&\left\|\widetilde{\mathbf{w}}^{(n+1)}-\mathbf{w}^*\right\|^2 \notag\\
        =&\mathbb{E}\left\|{\mathbf{w}}^{(n)}-\eta^{(n)}\sum_{m=1}^M\frac{d_m}{d}\widetilde{\mathbf{g}}_m^{(n)}-\mathbf{w}^*\right\|^2 \notag\\
        =&\mathbb{E}\left\|{\mathbf{w}}^{(n)}-\eta^{(n)}\sum_{m=1}^M\frac{d_m}{d}\nabla L_m\left(\mathbf{w}^{(n)}\right)-\mathbf{w}^*\right. \notag\\&\qquad\qquad\qquad~+\left.\eta^{(n)}\sum_{m=1}^M\frac{d_m}{d}\left(\nabla L_m\left(\mathbf{w}^{(n)}\right)-\widetilde{\mathbf{g}}_m^{(n)}\right)\right\|^2 \notag\\
        =&\underbrace{\left\|{\mathbf{w}}^{(n)}-\eta^{(n)}\sum_{m=1}^M\frac{d_m}{d}\nabla L_m\left(\mathbf{w}^{(n)}\right)-\mathbf{w}^*\right\|^2}_{A_1}\notag\\&\qquad~~+\underbrace{\left(\eta^{(n)}\right)^2\sum_{m=1}^M\frac{d_m^2}{d^2}\mathbb{E}\left\|\nabla L_m\left(\mathbf{w}^{(n)}\right)-\widetilde{\mathbf{g}}_m^{(n)}\right\|^2}_{A_2}.
        \end{align}
    First analyze the item $A_1$. 
    \begin{align}
        A_1 \notag
        =&\left\|{\mathbf{w}}^{(n)}-\eta^{(n)}\sum_{m=1}^M\frac{d_m}{d}\nabla L_m\left(\mathbf{w}^{(n)}\right)-\mathbf{w}^*\right\|^2\notag\\
        =&\left\|{\mathbf{w}}^{(n)}-\eta^{(n)}\nabla L\left(\mathbf{w}^{(n)}\right)-\mathbf{w}^*\right\|^2\notag\\
        =&\left\|{\mathbf{w}}^{(n)}-\mathbf{w}^*\right\|^2-2\eta^{(n)}\left<{\mathbf{w}}^{(n)}-\mathbf{w}^*,\nabla L\left(\mathbf{w}^{(n)}\right)\right>\notag\\&\qquad\qquad\qquad\qquad\qquad\qquad+\left(\eta^{(n)}\right)^2\left\|\nabla L\left(\mathbf{w}^{(n)}\right)\right\|^2\notag\\
        \leq &\left\|{\mathbf{w}}^{(n)}-\mathbf{w}^*\right\|^2+2\ell\left(\eta^{(n)}\right)^2\left(L\left({\mathbf{w}}^{(n)}\right)-L\left(\mathbf{w}^*\right)\right)\notag\\&-\mu\eta^{(n)}\left\|{\mathbf{w}}^{(n)}-\mathbf{w}^*\right\|^2 \notag
        -2\eta^{(n)}\left(L\left({\mathbf{w}}^{(n)}\right)-L\left(\mathbf{w}^*\right)\right)\notag\\
        =&\left(1-\mu\eta^{(n)}\right)\left\|{\mathbf{w}}^{(n)}-\mathbf{w}^*\right\|^2\notag\\&\qquad\qquad-2\eta^{(n)}\left(1-\ell\eta^{(n)}\right)\left(L\left({\mathbf{w}}^{(n)}\right)-L\left(\mathbf{w}^*\right)\right)\notag\\
        \leq &\left(1-\mu\eta^{(n)}\right)\left\|{\mathbf{w}}^{(n)}-\mathbf{w}^*\right\|^2\notag\\&\qquad\qquad\qquad\qquad-\mu\eta^{(n)}\left(1-\ell\eta^{(n)}\right)\left\|{\mathbf{w}}^{(n)}-\mathbf{w}^*\right\|^2 \notag\\
        \leq &(1-\frac{3\mu}{2}\eta^{(n)})\left\|{\mathbf{w}}^{(n)}-\mathbf{w}^*\right\|^2, \label{eq: A1 last equation}
    \end{align}
    where (\ref{eq: A1 last equation}) is due to $\eta^{(n)}\leq \frac{1}{2L}$. 

    Then analyze the item $A_2$. 
    \begin{align}
        A_2 &= \left(\eta^{(n)}\right)^2\sum_{m=1}^M\frac{d_m^2}{d^2}\mathbb{E}\left\|\nabla L_m\left(\mathbf{w}^{(n)}\right)-\widetilde{\mathbf{g}}_m^{(n)}\right\|^2 \notag\\
        &\leq \left(\eta^{(n)}\right)^2\sum_{m=1}^M\frac{d_m^2}{d^2} \sigma^2.  \label{eq: A2 last equation}
        \end{align}
    Combine (\ref{eq: A1 last equation}) and (\ref{eq: A2 last equation}), we can get the desired results, which ends the proof. 
\end{proof}
Substitute (\ref{eq: effect of outage}), (\ref{eq:effect of compression}), and (\ref{eq:effect of SGD}) into (\ref{eq: decompose of gap}), combining Assumption 4, Lemma \ref{eq: per round improvement} can be  obtained, which ends the proof.
\bibliographystyle{ieeetr}

\bibliography{BibDesk_File_v2}
\vspace{-2mm}    
\end{document}